\documentclass[11pt,fleqn]{article}

\usepackage{amsmath,amssymb,amsthm,fullpage,pgfplots,setspace,tikz,verbatim,ctable}
\usepackage[authoryear,round]{natbib}
\pgfplotsset{compat=newest}

\newcommand{\F}{\gamma}

\newtheorem{theorem}{Theorem}
\newtheorem*{theorem*}{Theorem}
\newtheorem{proposition}{Proposition}
\newtheorem{lemma}{Lemma}

\theoremstyle{definition}
\newtheorem{remark}{Remark}
\newtheorem{definition}{Definition}
\newtheorem{example}{Example}

\begin{document}

\title{\textbf{Successive Incentives}\thanks{We thank Peter Bogetoft for helpful comments. We are also grateful to seminar and conference participants at Stony Brook, Istanbul Technical University, Waseda University, and Chinese University of Hong Kong for valuable comments and suggestions. Jens Gudmundsson and Jens Leth Hougaard gratefully acknowledge financial support from the Carlsberg Foundation (grant no. CF18-1112). Moreno-Ternero acknowledges financial support from the Spanish Government through grant PID2020-115011GB-I00, funded by MCIN/AEI/10.13039/501100011033.
}}
\author{
\textbf{Jens Gudmundsson}%
\thanks{Department of Food and Resource Economics, University
of Copenhagen, Denmark.%email: jlh@ifro.ku.dk
}  \\
\textbf{Jens Leth Hougaard}%
\thanks{Department of Food and Resource Economics, University
of Copenhagen, Denmark.%email: jlh@ifro.ku.dk
}   \\
\textbf{Juan D. Moreno-Ternero}%
\thanks{Department of Economics, Universidad Pablo de Olavide, Spain.%; CORE, Universit\'{e} catholique de Louvain, Belgium.%, email: jdmoreno@upo.es%
}
\\ %\and
\textbf{Lars Peter \O sterdal}%
\thanks{Department of Economics, Copenhagen Business School,
Denmark.%, email: lpo.eco@cbs.dk%
}} 

\maketitle

\begin{abstract}
    We study the design of optimal incentives in sequential processes. To do so, we consider a basic and fundamental model in which an agent initiates a value-creating sequential process through costly investment with random success. If unsuccessful, the process stops. If successful, a new agent thereafter faces a similar investment decision, and so forth. For any outcome of the process, the total value is distributed among the agents using a reward rule. Reward rules thus induce a game among the agents. By design, the reward rule may lead to an asymmetric game, yet we are able to show equilibrium existence with optimal symmetric equilibria. We characterize optimal reward rules that yield the highest possible welfare created by the process, and the highest possible expected payoff for the initiator of the process. Our findings show that simple reward rules invoking short-run incentives are sufficient to meet long-run objectives. 
\end{abstract}

\textit{JEL} Classification: C70, L24, M52


\maketitle

\newpage

\section{Introduction}

Economists and social scientists alike have long been concerned with the study of incentives.
The focal question of how taxes affect incentives to work has constantly been in the spotlight \cite[e.g.,][]{break1957income}, but attention has also long been paid to the study of incentives in diverse settings, such as school work or marketing \cite[e.g.,][]{hurlock1925evaluation, haring1953special}.
A central impetus was given decades ago to the formulation of mechanisms to induce team agents to make decisions conforming to the organizational interest \cite[e.g.,][]{alchian1972production, groves1973incentives}, with a special emphasis on the informational aspects of the problem \cite[e.g.,][]{holmstrom1982moral}.
But the search for optimal incentives for collective intelligence 
is of increasing importance nowadays, as global systems of communication, governance, trade, and transport grow rapidly in complexity \cite[e.g.,][]{mann2017optimal}.
Most of these global systems are sequential, as many other processes that abound in nature.
For instance, sequentiality is a pervasive feature of the organization of natural and artificial systems \cite[e.g.,][]{corominas2013origins} and it frequently appears in large scale industrial manufacturing, clinical trials, decentralized computer networks, or decision making and planning, among others \cite[e.g.,][]{azoulay2004capturing, antras2013organizing, winter2006optimal}.
Finally, mining in blockchains, social mobilization, auction design or multi-level marketing have also recently become popular instances of sequential processes \cite[e.g.,][]{pickard2011time, huberman2019economist, li2022diffusion}.

The aim of this paper is to study \textit{successive incentives}.
That is, incentives in a sequential process in which agents make concatenated decisions.
To do so, we consider a stylized model in which an agent (the initiator) initiates a value-creating sequential process with random success by investing money (or effort). If the investment is unsuccessful the process stops; if it is successful the process continues and a new agent faces a similar investment decision, and so on. 
The success rate is common to all agents, and increasing in the investment, but bounded above too.
That is, the success of an investment is never fully guaranteed (even for arbitrary large amounts).
The initiator arrives with a value and each time the process is extended additional value is created. We assume this value is the same each time (and thus normalize it to 1).
Our aim is to explore optimal \textit{reward rules} that specify how the overall generated value in the process is distributed, for any realization of agents (that is, for any situation where the last agent is unsuccessful and thus terminates the process).
As such, each reward rule induces a game among agents, who will invest strategically.
Our first result, Theorem 1, states that, for any choice of rule, there always exists an equilibrium for such a game.
This relies on the fact that the game exhibits an interesting feature: each agent's best response is bounded above, independently of the other agents' actions.
The bounds are nevertheless individual and unboundedly increasing.
In other words, there is no common bound for all agents' best responses, and we can find equilibria with unbounded investments.
The underlying rationale is that the value created by earlier agents can be used to generate even larger incentives for later agents (without imposing any particular structure on the success rate).

Equilibria can also exhibit interesting features regarding the comparison of agents' investment decisions with the first-best investment (the one that maximizes the expected overall generated value in the process).
For instance, we show that, with a mild condition on the success rate, there is an equilibrium in which every agent except the initiator invests more than the first-best amount, whereas the initiator in any equilibrium invests less than the first-best amount.

Nevertheless, as mentioned above, our main research question is to obtain optimal rules to manage the \textit{successive incentives}.
We consider two natural optimality notions in our setting.
On the one hand, to induce investments yielding the highest possible overall expected value of the (sequential) process.
On the other hand, to induce investments yielding the highest possible expected payoff for the initiator.

Regarding the first notion, our Theorem~\ref{TH:socialopt} shows that there is a unique investment profile maximizing the overall expected welfare of the (sequential) process which can be supported in equilibrium.
Such a profile is constant, i.e., all agents invest the same amount. 
This leads to under-investment compared to the first-best investment level.
Multiple rules support that profile (in equilibrium), but a canonical example is the rule in which the value created is split equally among all successful agents. This equal-split rule has the practical advantage that it does not rely on information about the success rate. Therefore, an external benevolent planner can use this rule without knowledge of the success rate and still be sure to implement the unique socially optimal investment profile in equilibrium.

Regarding the second notion of initiator optimality, our Theorem~\ref{TH:initiator} shows that there is a unique investment profile maximizing the expected payoff for the initiator. Such a profile is \textit{near-constant}, i.e., each agent in the process, except for the initiator, invests the same amount. The initiator invests a higher amount. As with Theorem~\ref{TH:socialopt}, multiple rules support that profile (in equilibrium). But it contrast with that case, now all the supporting rules depend on the success rate. Now, as it is common knowledge among the agents, the initiator can easily construct the appropriate rule for implementation. A focal example is the reward rule imposing a specific fixed rate for all successful agents, and leaving the rest for the initiator.

We also consider an extension of the model where agents' investments are limited by the value that is generated through the process itself. In other words, the entire value might not be used to incentivize agents. A key aspect now for a supporting rule is that agents have to be rewarded a positive amount even if they are unsuccessful (to finance their investment) and this amount must be ``set aside" for the agent going forward.  
This changes the socially optimal investment profile. More precisely, our Theorem~\ref{TH:SF-eq} shows that the initiator will now invest more than the rest of the agents, who will all invest a (common) smaller amount (in contrast to the fully constant socially optimal profile we obtained without budget constraints). Theorem~\ref{TH:SF-eq} includes the program that determines the optimal investment profile and presents a simple rule along the lines of the initiator optimal rule, where now some resources are set aside for investments.

Common to all our results on optimal equilibria is that they are implementable by \textit{simple} rules: the value generated by any successful agent is split between this agent and the initiator. This is quite remarkable. By design we could have chosen reward rules with much more complicated structures, treating agents asymmetrically. For instance, rules that incentivize agents via payments conditional on reaching some future state (as in various types of bonus schemes). Such rules may also support the optimal equilibria, yet we demonstrate that complexity is not needed to achieve the long-run goals of incentivization.

The rest of the paper is organized as follows.
We first conclude this introduction providing the connections of our work with the existing literature.
In Section~\ref{SEC:model}, we set up the model and main basic concepts we use throughout the paper. 
In Section~3, we begin our analysis of successive incentives.
We first show that the game induced by reward rules guarantees existence of equilibrium and also emphasize some interesting features of those equilibria. We further provide a few structural results that are key to our subsequent analysis of optimal reward rules. In Section~4 we obtain our main results concerning socially 
 optimal and initiator optimal equilibrium profiles and the reward rules that supports them.
In Section~5, we provide an extension of our analysis to deal with the case in which agents have budget constraints that may limit their investment decisions.
Finally, Section~6 concludes.
To ease exposition, we defer some of our proofs to an appendix.

\paragraph{Related literature}

Our work could be considered as part of the emerging literature on \textit{induction network interventions}, which stimulate peer-to-peer interaction to create cascades in behavioral diffusion, implicitly endorsing that secondary incentives can be more efficient and effective than primary incentives, at least in some contexts \cite[e.g.,][]{ballester2006s, valente2012network, galeotti2020targeting}.
Our analysis of optimal reward rules is actually reminiscent of the analysis in \citet{galeotti2021taxes}, who design and evaluate welfare-enhancing tax policy schemes in complex supply chains consisting of primary and final good producers.
In a similar vein, \citet{bimpikis2019cournot} consider a model of competition among firms that produce a homogeneous good in a networked environment, in which a bipartite graph determines which subset of markets a firm can supply to.
Therein, the social planner has different gains from intervening in different parts of the network.
This is also the case in \citet{elliott2022supply}, where they study novel equilibrium fragilities in the strategic formation of large supply networks, along with new methods for analyzing them.

Our model could also be considered as a moral hazard optimal contracting problem, as pioneered by \citet{holmstrom1982moral}, but in a dynamic setting \cite[e.g.,][]{melumad1995hierarchical, mookherjee2006decentralization}.
In such a setting, followers may shirk when effort is not observed by their predecessors.
This leads predecessors to monitor the effort exerted by their immediate followers.%
\footnote{\citet{gershkov2015formal} show that if the technology satisfies complementarity, peer monitoring substitutes for the principal's monitoring. 
\citet{halac2022monitoring} further study the optimal monitoring of teams, together with the scheme of performance-contingent rewards.}
The final output of the sequential process is determined by a production function, which is cumulative in the efforts of workers and managers at all levels \cite[e.g.,][]{qian1994incentives}.
Within this framework, \citet{winter2006optimal} characterizes the mechanism to allocate rewards among agents so as to induce all of them to exert effort in equilibrium at minimal cost to the principal.
A crucial feature of such a mechanism is to discriminate among agents depending on the degree their efforts are unobservable by their peers.%
\footnote{\citet{winter2004incentives} analyzes the counterpart model with simultaneous decisions, also obtaining that the optimal investment-inducing mechanism endorses discrimination among equals.
\citet{moriya2020asymmetric} 
and \citet{halac2021rank} consider related problems in which agents' incentives to work depend on a hidden state of nature.} 

But the closest contribution to our work is \citet{hougaard2022optimal}.
The focus of that paper is to study the optimal management of evolving hierarchies by means of a similar (more specific) model in which reward rules only allow for upward transfers in the hierarchy.
They also obtained a unique investment profile maximizing the overall expected value of the hierarchy, as well as multiple optimal schemes with respect to the initiator's payoff.
Common to all those results is the prominent role of constant investments for all agents following the initiator along the hierarchy.
In that sense, our results here provide robustness to some of the results in \citet{hougaard2022optimal}.
However, our analysis here is much more general, considering sequential processes for which incentives can be broadly defined, without restricting reward rules to impose only upward transfers.
Specifically, this results in an important difference between the game that is induced by the respective reward rules.
The game in \citet{hougaard2022optimal} has a particularly convenient structure being a supermodular game,%
\footnote{A game is \textit{supermodular} when the marginal value of a player's action is increasing in the other players' actions.
It is well known that for continuous and supermodular game with compact intervals for strategy spaces, agents have increasing best-response functions and, thus, there always exists at least one Nash equilibrium in pure strategies \cite[e.g.,][]{topkis1979equilibrium, milgrom1990rationalizability}.}
whereas the induced game in the present analysis is more complex. 

We also connect to the literature dealing with resource allocation in the presence of a network structure \cite[e.g.,][]{littlechild1973simple, myerson1977graphs, megiddo1978computational, hougaard2017sharing, juarez2018sharing, Hou18}.
Within this literature, the structure of the network might be exploited to define fair allocation among agents connected in the graph, and the minimal distance to the root becomes crucial (as it represents the stand-alone option for the agents).
Sometimes, networks are just assumed to restrict cooperation.
And, some other times, fairness requirements are related directly to the hierarchical network structure. 

We  also relate to an emerging literature dealing with blockchains \cite[e.g.,][]{leshno2020bitcoin, huberman2021monopoly, prat2021equilibrium}. In Proof-of-Work protocols, such as Bitcoin, miners compete to verify new blocks when extending the blockchain. The probability of winning the right to verify a block is proportional to the miner’s computational power in the network. Thus, mining is costly and the miner is rewarded by new Bitcoins (as well as transaction fees users pay). The blockchain itself can therefore be construed as a dynamic process where miners invest costly efforts (solving cryptographic puzzles) with random success. If successful, a value (the block reward) is obtained. The way block rewards are paid out influences the incentives to mine. Our results can shed light on the optimal design of rewards within consensus protocols such as Proof-of-Work.   

\section{Preliminaries} \label{SEC:model}

Successive investment decisions are ubiquitous in society and come in many forms.
A key aspect in these decisions is that costly actions taken by some (economic) agents enable, or force, other agents to take costly actions themselves. The purpose of the investments can either be to create additional value or to mitigate harm.  We shall focus on the former case, where predecessor groundwork is essential for successor progress, and thus we endorse the principle that one ``stands on the shoulders'' of others. This feature is often experienced in diverse settings. Our running examples will be the case of new discoveries building on previous insights in R\&D, as well as the case of investments over time to maintain a value-generating resource (such as an ecosystem creating societal value) but requiring regular costly pollution abatement to thrive.
In either of those cases, the efficient \textit{incentivization} of agents’ investments in the process is a crucial aspect for the practical viability of the process.

Formally, there is an \textbf{initiator} of a sequential process, agent $0$, and an infinite \textbf{set of potential agents} $\mathbb{N} = \{ 1, 2, \dots \}$; let $\mathbb{N}_0 \equiv \{0\} \cup \mathbb{N}$.
The initiator can add value by extending the process through costly investment with random success. If unsuccessful, the process stops.
If successful, a new agent $1$ thereafter faces a similar investment decision, and so forth.
The initiator starts out with a value and each time the process is successfully extended a new value is generated: we assume this value to be the same each time and normalize it to 1.  

An investment profile, in short a \textbf{profile}, is a collection of agents' investments, i.e., $x = (x_0, x_1, \dots)$, where $x_i \geq 0$ for every agent $i$. 
Let $x_{>i} \equiv (x_{i+1}, x_{i+2}, \dots)$ and $x_{\geq i} \equiv (x_i, x_{>i})$.
A profile $x$ is \textbf{constant} if $x = (c, c, \dots)$, \textbf{near-constant} if $x_{\geq 1}$ is constant, and \textbf{constant-tail} if $x_{\geq i}$ is constant for some $i$.
The \textbf{set of profiles} is $X \equiv \mathbb{R}_{\geq 0}^{\mathbb{N}_0}$.

The higher the investment, the more likely the agent is successful in extending the process, but there is always a risk of failure.
Specifically, there is $\varepsilon > 0$ such that, with investment $x_i \geq 0$, agent $i$ successfully extends the process with probability $p(x_i) \in [0,1-\varepsilon]$ with $p(0) = 0$.
The \textbf{success rate} $p: \mathbb{R}_{\geq 0} \to [0, 1 - \varepsilon] $ is common to all agents and fixed throughout. 
Moreover, $p$ is assumed to be increasing, differentiable, and concave in the investment and increases steeply at zero, that is, $p'(x_i) \to \infty$ as $x_i \to 0$.
We also assume that $p$ is such that the ratio 
$\frac{p(\cdot)} { p'(\cdot)}$ is a convex function.\footnote{This ratio is also defined in a different setting as the so-called \textit{fear of ruin} \cite[e.g.,][]{aumann1977power}.}

The value of the process is shared among agents in order to incentivize investments.
Formally, a reward rule $f$, in short a \textbf{rule}, specifies how the value is distributed for any realization of agents.
In particular, for each $k \in \mathbb{N}_0$, agent $i \in \{0, \dots, k\}$ is assigned $f(i,k) \geq 0$ of the total generated value, $k+1$.
Moreover, rules are \textit{balanced} in the sense that $f(0,k) + \dots + f(k,k) = k+1$. 
Let $F$ be the resulting \textbf{set of rules}. Formally,
\[
    F = \{ f \colon \mathbb{N}_0 \times \mathbb{N}_0 \to \mathbb{R}_{\geq 0} 
    \text{ such that }f(0,k) + \dots + f(k,k) = k+1 \text{ for each }k \in \mathbb{N}_0\}.
\]
It is useful to represent rules in matrix form. To do so, we consider $f(i,k)$ in column $i = 0, 1, \dots$ and row $k \geq i$. In this way, the $i$th column is the potential ``payoff stream'' that agent $i$ faces when making her investment decision, whereas the $k$'th row represents the values allocated to the successful agents $\{0,\dots, k-1\}$ and the unsuccessful agent $k$ who terminates the process.\footnote{Note that the $k$'th row is realized with probability $\prod_{i=0}^{k-1}p(x_i)(1-p(x_k))$.}
This is illustrated in Example~\ref{EX:rules}. 

\begin{example}[Examples of rules] \label{EX:rules}
    The canonical \textit{equal split} solution shares the total value generated by successful investment equally among all successful agents.
    This is represented by the rule $f^{ES}$ below.
    Agent $i > 0$ only receives a positive payoff once successful (row $i+1$ and on).
    The \textit{equal split} rule is a distinguished member of a parametric family we dub \textit{fixed-fraction rules}\footnote{See also \cite{hougaard2022optimal} and \cite{GudmundssonHougaardKo2023}.} and denote by $f^{\alpha}$.
    Here, all successful agents get a fixed fraction $\alpha \in [0,1]$ with the residual value going to the initiator.
    Clearly, $f^{ES}= f^1$.
    \[
        f^{ES} = \left[ \begin{array}{cccc} 
            1 \\
            2 & 0 \\
            2 & 1 & 0 \\
            2 & 1 & 1 & 0 \\
            \vdots
        \end{array} \right] 
        \quad
        f^\alpha = \left[ \begin{array}{cccc}
            1 \\
            2 & 0 \\
            3 - \alpha & \alpha & 0 \\
            4 - 2\alpha & \alpha & \alpha & 0 \\
            \vdots
        \end{array} \right]
    \]
    
    We do not impose any structure on how the rewards change as the process is extended. In particular, this allows to accommodate rules such as $f^{JP}$, in which an ever-increasing ``jackpot'' is awarded to the agent who is last to succeed in extending the process.
    Under such a rule, each agent $i$ wants to succeed but prefers $i+1$ to fail (however, they can only affect their own success).
    \[
        f^{JP} = \left[ \begin{array}{cccc} 
            1 \\
            2 & 0 \\
            1 & 2 & 0 \\
            1 & 0 & 3 & 0 \\
            \vdots
        \end{array} \right] 
        \quad
        f^{\alpha,\F} = \left[ \begin{array}{cccc} 
        1 \\
        2 - \F & \F \\
        3 - \alpha - \F & \alpha & \F \\
        4 - 2 \alpha - \F & \alpha & \alpha & \F \\
        \vdots
    \end{array} \right] 
    \]
    This illustrates also the asymmetry between agents:
    as we get further into the process, new ``payoff streams'' become available.
    For instance, comparing agents $0$ and $1$, we have $f(0,0) = 1$ whereas $f(1,1) \in [0,2]$;
    we have $f(0,1) \in [0,2]$ but $f(1,2) \in [0,3]$.
    In this way, we can incentivize agent $1$ to make investments that are not possible to incentivize for agent $0$.
    Finally, $f^{\alpha, \F}$ represents a variation of the class of fixed-fraction rules where unsuccessful agents $i>0$ receive positive payment $\F \in [0,2]$.
    $\hfill \circ$
\end{example}

Given profile $x$, the \textbf{expected investment} is
\[
    \mathbb{I}(x) \equiv x_0 + \sum_{j=1}^{+\infty} \prod_{i=0}^{j-1} p(x_i) x_j. 
\]
As each realized agent adds value $1$, the \textbf{expected value} of the process is
\[
    \mathbb{V}(x) \equiv 1 + \sum_{j=1}^{+\infty} \prod_{i=0}^{j-1} p(x_i).
\]
As $p$ is bounded, the expected value is bounded but the expected investments might not be so.%
\footnote{\label{FN:bounded}Formally, as $p(x_i) \leq 1 - \varepsilon$, $\mathbb{V}(x) \leq 1 + (1 - \varepsilon) + (1 - \varepsilon)^2 + \dots = 1 / \varepsilon$.
On the other hand, when agent $0$ invests $x_0 = 1$ and agent $j > 0$ invests $x_j =\frac{1}{  (j+1) \prod_{i < j} p(x_i) }$, 
we have $\mathbb{I}(x) = 1 + 1/2 + 1/3 + \dots$, which is the divergent harmonic series.}
When expected investments are bounded, we define the \textbf{expected welfare} as 
\[
    \mathbb{W}(x) \equiv \mathbb{V}(x) - \mathbb{I}(x).
\]

Given a rule $f$, each agent $i$ maximizes the expected payoff $U_i(x_{\geq i},f) \in \mathbb{R}$, which consists of three parts: the amount $f(i,i) \geq 0$ assigned if the investment is unsuccessful; the amount $R_i(x_{>i},f) \geq 0$ expected to be assigned if successful; the costly investment $x_i \geq 0$. That is,
\[
    U_i(x_{\geq i},f) \equiv (1 - p(x_i)) f(i,i) + p(x_i) R_i(x_{>i}, f) - x_i,
\]
where
\[
    R_i(x_{>i}, f) \equiv \sum_{k = i+1}^{+ \infty} \prod_{j = i+1}^{k-1} p(x_j) (1 - p(x_k)) f(i,k).
\]
For instance, $R_0(x_{>0}, f) = (1 - p(x_1)) f(0,1) + p(x_1) (1 - p(x_2)) f(0,2) + \dots$.

As the rules are balanced, we obtain the following identity:
\[
    U_0(x_{\geq 0},f) + \sum_{j = 1}^{+ \infty} \prod_{i = 0}^{j-1} p(x_i) U_j(x_{\geq j},f) = \mathbb{W}(x).
\]

Each rule $f$ therefore induces a game where the agents choose their investment levels strategically.
In this game, the rule $f$ and the success rate $p$ are common knowledge, whereas an agent's investment cannot be observed by the remaining agents.
Hence, later agents cannot condition their choice on earlier investments.

\begin{definition}
    A rule $f$ is said to \textbf{support} the profile $x$ whenever $x$ is a Nash equilibrium in the game induced by $f$.
    That is, for each agent $i$ and investment $y_i \geq 0$,
    \[
        U_i((x_i,x_{>i}), f) \geq U_i((y_i,x_{>i}), f). 
    \]
    A profile $x$ is an \textbf{equilibrium profile} if supported by some rule. 
    \hfill $\circ$
\end{definition}

For each $x_i\ge 0$, let $g(x_i)= p(x_i) / p'(x_i)$. Recall that we assumed $p$ to be such that $g$ is convex. For convenience, let $g(0) \equiv \lim_{x_i \to 0} p(x_i) / p'(x_i) = 0$.
As $p$ is concave, $g(x_i) / p(x_i) = 1 / p'(x_i)$ is increasing and $p'(x_i) \cdot (x_i - 0) < p(x_i) - p(0)$. Thus, $g(x_i) > x_i$ for each $x_i > 0$.
Lemma~\ref{OBS:individual} follows immediately from maximizing $U_i(x_{\geq i}, f)$ with respect to $x_i$.

\begin{lemma}
\label{OBS:individual}
    The following statements are equivalent:
    \begin{itemize}
        \item The rule $f$ supports the profile $x$.
        \item For each agent $i$, if $R_i(x_{>i},f) \leq f(i,i)$, then $x_i = 0$;
        otherwise, $x_i > 0$ is such that 
        \[
            R_i(x_{>i}, f) - f(i,i) = \frac{g(x_i)}{p(x_i)}.
        \]
    \end{itemize}
\end{lemma}

By Lemma~\ref{OBS:individual}, the expected equilibrium payoff is $U_i(x_{\geq i},f) = f(i,i) + g(x_i) - x_i$.
That is to say, $g(x_i)$ can be viewed as the gross (excluding investment $x_i$) expected gain over $f(i,i)$ for which investing $x_i$ is optimal. 
Naturally, incentivizing higher investments requires higher returns; 
that is, $g$ is increasing.
We take this one step further as $g$ is convex:
it is increasingly more expensive to incentivize higher investments.%
\footnote{There is a parallel here to the different but related setting of risk-averse gambling.
A risk-averse agent with utility function $u$ pays price $x_i$ for a lottery ticket that awards a prize with probability $\pi$. 
Then the minimal prize amount $g(x_i)$ required for the agent to take part is such that $(1 - \pi) u(- x_i) + \pi u(g(x_i) - x_i) = u(0)$.
As $u$ is concave, $g$ must be convex.
In our setting, the agents are risk-neutral (corresponding to linear $u$) whereas the success rate (corresponding to $\pi$) is concave.}

Lemma~\ref{OBS:individual} pertains to individual equilibrium payoffs.
It will also be useful to have an aggregate expression. 
For that purpose, let
\[
    \mathbb{\hat{V}}(x,f) \equiv f(0,0) + \sum_{j=1}^{+\infty} \prod_{i=0}^{j-1} p(x_i) f(j,j) + \mathbb{G}(x),
\]
where 
\[
    \mathbb{G}(x) \equiv g(x_0) + \sum_{j=1}^{+\infty} \prod_{i=0}^{j-1} p(x_i) g(x_j).
\]
Intuitively, $\mathbb{G}(x)$ can be viewed as the expected (aggregate) cost of incentivizing $x$.
As, in general, there might be both over- and underinvestment at a profile $x$ for a given rule $f$, there is no systematic relation between $\mathbb{V}$ and $\hat{\mathbb{V}}$.
However, in equilibrium, Lemma~\ref{OBS:aggregate} shows that they are equal.%
\footnote{The converse of Lemma~\ref{OBS:aggregate} is not true.
For instance, we may have $\mathbb{V}(x) - 1 \geq \mathbb{G}(x)$ and yet $x$ cannot be supported (see Example~\ref{EX:nearconstantNotEq}).}

\begin{lemma}
\label{OBS:aggregate}
    If the rule $f$ supports the profile $x$, then $\mathbb{V}(x) = \hat{\mathbb{V}}(x,f)$.
    Hence, for each equilibrium profile $x$, $\mathbb{V}(x) - 1 \geq \mathbb{G}(x)$.
\end{lemma}

\begin{proof}
    As noted after Lemma~\ref{OBS:individual}, for each agent $i$,
    \[
        U_i(x_{\geq i}, f) = f(i,i) + g(x_i) - x_i.
    \]
    Aggregating over all agents,
    \[
        U_0(x_{\geq 0},f) + \sum_{j = 1}^{+ \infty} \prod_{i = 0}^{j-1} p(x_i) U_j(x_{\geq j},f)
        =
        f(0,0) + g(x_0) - x_0 + \sum_{j=1}^{+\infty} \prod_{i=0}^{j-1} p(x_i) ( f(j,j) + g(x_j) - x_j ).
    \]
    As noted, as the rules are balanced, the left-hand side is identical to $\mathbb{W}(x)=\mathbb{V}(x) - \mathbb{I}(x)$;
    the right-hand side is, by definition, $\hat{\mathbb{V}}(x,f) - \mathbb{I}(x)$.
    The first part of the statement then follows.
And the second part then follows from $f(0,0) = 1$ and $f(j,j) \geq 0$:
    \[
        \mathbb{\hat{V}}(x,f) 
        = f(0,0) + \sum_{j=1}^{+\infty} \prod_{i=0}^{j-1} p(x_i) f(j,j) + \mathbb{G}(x)
        \geq 1 + \mathbb{G}(x).
        \qedhere
    \]
\end{proof}
The following Remark~\ref{REM:constant} illustrates the implications of Lemma~\ref{OBS:aggregate} for the particular case of constant profiles.

\begin{remark}[Supporting constant profiles] \label{REM:constant}
    For a constant equilibrium profile, $\bar{x} = (c, c, \dots)$, Lemma~\ref{OBS:aggregate} reduces to $p(c) \geq g(c)$, as 
    \[
        \mathbb{V}(\bar{x}) - 1 
        = \frac{p(c)}{1 - p(c)} 
        \geq \frac{g(c)}{1 - p(c)}
        = \mathbb{G}(\bar{x}).
    \]
    Indeed, $p(c) \geq g(c)$ turns out to be both necessary and sufficient. One can simply consider the rule $f^{\alpha, \F}$, with $\alpha = 1$ and $\F= 1- \frac{g(c)}{p(c)}$,  to support $\bar{x}$.\footnote{Note that when $p(c) = g(c)$, the rule reduces to $f^{ES}$.}
\end{remark}

\section{Structural results} \label{SEC:structural}

In this section, we address some fundamental ideas that set the stage for the ensuing analysis. 
First, we deal with the essential aspect throughout the paper of equilibrium existence.
In Subsection~\ref{SUB:existence}, Theorem~\ref{TH:equilibriumExistence} shows that every rule $f$ supports some investment profile $x$ by invoking an existence theorem due to \citet{Ma1969}.
Second, we explore cost-effective investment profiles.
In Subsection~\ref{SUB:constanttail}, we show that a ``flatter'' profile is both cost effective and cheaper to incentivize.
Specifically, let $x$ and $\bar{x} = (x_0, \dots, x_{k-1}, c, c, \dots)$ be such that $\mathbb{V}(x) = \mathbb{V}(\bar{x})$;
then $\mathbb{I}(x) > \mathbb{I}(\bar{x})$ (Proposition~\ref{PR:flatten}) and $\mathbb{G}(x) > \mathbb{G}(\bar{x})$ (Proposition~\ref{PR:cheaper}). These features will become crucial for our main results in Section~\ref{SEC:optimality}.

\subsection{Equilibrium existence} \label{SUB:existence}

The domain of rules is extremely rich and different rules give rise to very different incentive structures.
For instance, if we consider two of the rules from Example~\ref{EX:rules}, we can observe that an agent's response to an increased investment by another (successor) agent might be qualitatively different. More precisely, let $i$ and $j$ be such that $j > i$. Then, 
under $f^{ES}$, $i$'s optimal investment is unchanged, whereas under $f^{JP}$, it decreases.
Hence, our induced game need not be supermodular.\footnote{As mentioned in the Introduction, this is in contrast with the model analyzed in \citet{hougaard2022optimal}.}

The asymmetry between the agents allows equilibrium investments to grow unboundedly, as the value created by earlier agents can be used to generate even larger investment incentives for later agents.
That is to say, even though a successful investment always adds value $1$, we can have equilibrium investments exceeding $1$---indeed, even though the expected value of the process in its entirety is bounded (footnote~\ref{FN:bounded}), an agent may invest \emph{even more} than that in equilibrium.
We return to this in Subsection~\ref{SUB:constanttail} where we will find cases of overinvestment by \emph{everyone} except the initiator.
While this does not immediately indicate a problem for equilibrium existence, it shows that we cannot limit our search to bounded profiles (for instance, by $1$ or even by $\mathbb{V}(x)$), which at first glance would appear natural.
And, even more, the above raises the concern that some agents might simply be always better off investing more, so there is no well-defined optimal investment.

\begin{example}[Unbounded equilibrium investments] \label{EX:unboundedEq}
    Let $x$ be a profile supported by the rule $f = f^{JP}$ as defined in Example~\ref{EX:rules}.
    Consider agent $i$.
    By Lemma~\ref{OBS:individual} and the design of $f$,
    \[
        \frac{g(x_i)}{p(x_i)} 
        = R_i(x_{>i},f) - f(i,i)
        = (1 - p(x_{i+1})) (i+1).
    \]
    By contradiction, suppose $x$ is bounded by some $\mathcal{B} \geq 0$:
    that is, for each agent $i$, $x_i \leq \mathcal{B}$.
    As $g/p$ is increasing, $g(\mathcal{B}) / p(\mathcal{B}) \geq g(x_i) / p(x_i)$.
    As $p$ is increasing and bounded, $p(x_{i+1}) \leq p(\mathcal{B}) \leq 1 - \varepsilon$, where $\varepsilon > 0$.
    Hence,
    \[
        \frac{g(\mathcal{B})}{p(\mathcal{B})} \geq \frac{g(x_i)}{p(x_i)} \geq \varepsilon \cdot (1 + i).
    \]
    This is a contradiction as the left-hand side is finite whereas the right-hand side grows unboundedly in $i$.
    \hfill $\circ$
\end{example}

Although Example~\ref{EX:unboundedEq} shows that there is no (common) upper bound on equilibrium investments, we are still able to identify \emph{individual} upper bounds.  That is to say, for a given success rate $p$ and agent $i$, there is a bound $B_i \geq 0$ such that, at every equilibrium profile $x$, we have $x_i \leq B_i$.
Even though these bounds $B_i$ are unboundedly increasing in $i$, each is still finite and the bounds can be used as a stepping stone to show equilibrium existence.

\begin{theorem} \label{TH:equilibriumExistence}
    For each rule $f$, there exists a profile $x$ such that $f$ supports $x$.
\end{theorem}

\begin{proof}
As $g/p$ is increasing, unbounded, and continuous, we can define $B_i \geq 0$ through 
    \[
        \frac{g(B_i)}{p(B_i)} = i + 1 + \frac{1}{\varepsilon}.
    \]
    
    Suppose that agent $i$ is choosing the best response to what all other agents are investing. Then by Lemma~\ref{OBS:individual}, and as $f(i,i) \geq 0$,
    \[
        \frac{g(x_i)}{p(x_i)} = R_i(x_{>i}, f) - f(i,i) \leq R_i(x_{>i},f).
    \]
    An upper bound on $R_i(x_{>i},f)$ is all ``past'' value created by agents $0$ through $i$ (which is $i+1$) and all ``future'' value created by $i+1$ and followers (which is $\mathbb{V}(x_{>i}) < 1 / \varepsilon$; see footnote~\ref{FN:bounded}).
    Hence, 
    \[
        \frac{g(x_i)}{p(x_i)} \leq i + 1 + \frac{1}{\varepsilon} = \frac{g(B_i)}{p(B_i)}.
    \]
    As $g/p$ is increasing, it follows that $x_i \leq B_i$.
    
    Let $X_i \equiv [0,B_i] \subset \mathbb{R}$ be the potential best responses for agent $i$.
    Let $u_i(x) \equiv U_i(x_{\geq i},f) = f(i,i) + p(x_i) ( R_i(x_{>i}) - f(i,i) ) - x_i$.
    As $p$ is concave in $x_i$, $u_i$ is concave (and thus quasi-concave) in $x_i$.
    As both $p$ and $R_i$ are continuous in $x$, so is $u_i$. 
    Equilibrium existence now follows from \citet[Theorem 4]{Ma1969}, restated here with $I \equiv \mathbb{N}_0$ being the set of agents:
    
    \begin{quotation} \noindent
        Let $\{X_i\}_{i \in I}$ be an indexed family, finite or infinite, of nonempty compact convex sets each in a separated topological vector space.
        Let $\{u_i\}_{i \in I}$ be a family of real-valued continuous functions defined on $X = \prod_{j \in I} X_j$.
        If for each $i \in I$ and for any fixed 
        \[
            x_{-i} \in \prod_{\substack{j \in I \\ j \neq i}} X_j,
        \]
        $u_i(x_i,x_{-i})$ is a quasi-concave function of $x_i \in X_i$, then there exists a point $y \in X$ such that for any $i \in I$,
        \[
            u_i(y) = \max_{z_i \in X_i} u_i(z_i,y_{-i}),
        \]
        where $y_{-i}$ is the projection of $y$ in
        \[
            \prod_{\substack{j \in I \\ j \neq i}} X_j.
        \]
    \end{quotation}
    For our purposes, the point $y$ is an equilibrium profile under $f$.
    This completes the proof.
\end{proof}

A canonical class of rules are those in which an agent's reward is affected by their own success but not the success or failures of others.
That is, $f$ is such that $f(i,i+1) = f(i,i+2) = \dots$;
see for instance $f^{ES}$ in Example~\ref{EX:rules}.
For such rules, $R_i(x_{>i},f) = f(i,i+1)$ and Lemma~\ref{OBS:individual} implies that $i$'s investment is the same in every equilibrium.
Moreover, even if this only applies to all agents $i > 0$ (compare $f^{\alpha,\gamma}$ in Example~\ref{EX:rules}), the conclusion extends also to agent $0$ and therefore implies that the equilibrium profile is unique.
In Sections~\ref{SEC:optimality} and~\ref{SEC:budget}, we will design rules under various optimality criteria;
all of the rules that we identify will be as above, meaning that they guarantee unique equilibrium investments. 

\subsection{Constant-tail profiles} \label{SUB:constanttail}

Two optimality criteria that we will explore pertain to maximizing welfare (subject to equilibrium constraints).
Given the identity $\mathbb{W}(x) = \mathbb{V}(x) - \mathbb{I}(x)$, a necessary condition for such a maximizer $x$ will be that it minimizes investments $\mathbb{I}(\cdot)$ for the particular level of value created $\mathbb{V}(x)$. 
Hence, as a stepping stone towards identifying optimal rules and investments, we first ask the following:
given a profile $x$, is there a way to reduce investments $\mathbb{I}$ without affecting the value created $\mathbb{V}$?

Our next result answers the previous question affirmatively.
Specifically, it suggests to ``flatten'' the tail of $x$ at any point to create $\bar{x} = (x_0, \dots, x_{k-1}, c, c, \dots)$, where $c \geq 0$ is chosen such that $\mathbb{V}(x) = \mathbb{V}(\bar{x})$.
Proposition~\ref{PR:flatten} states that this reduces the expected investments.
The assumption that $p$ is concave is essential for this result.

\begin{proposition}[Constant-tail profiles are cost effective] \label{PR:flatten}
    For each $x \in X$, each $k \in \mathbb{N}_0$, and each $\bar{x} = (x_0, \dots, x_{k-1}, c, c, \dots) \in X$ such that $\mathbb{V}(x) = \mathbb{V}(\bar{x})$, we have $\mathbb{I}(x) \geq \mathbb{I}(\bar{x})$.
\end{proposition}

\begin{proof}
    Say first $k = 0$, which corresponds to $\bar{x} = (c, c, \dots)$.
    As $p$ is concave, 
    \begin{align*}
        p \left ( \frac{\mathbb{I}(x)}{\mathbb{V}(x)} \right )
        &= p \left ( \frac{x_0 + p(x_0) x_1 + p(x_0) p(x_1) x_2 + \dots}{1 + p(x_0) + p(x_0) p(x_1) + \dots} \right ) \\
        &\geq \frac{p(x_0) + p(x_0) p(x_1) + p(x_0) p(x_1) p(x_2) + \dots}{1 + p(x_0) + p(x_0) p(x_1) + \dots}
        = \frac{\mathbb{V}(x) - 1}{\mathbb{V}(x)}.
    \end{align*}
    As $\mathbb{V}(x) = \mathbb{V}(\bar{x})$,
    \[
        p \left ( \frac{\mathbb{I}(x)}{\mathbb{V}(x)} \right )
        \geq \frac{\mathbb{V}(\bar{x}) - 1}{\mathbb{V}(\bar{x})}
        = p(c).
    \]
    As $p$ is increasing, we have $\mathbb{I}(x) / \mathbb{V}(x) \geq c = \mathbb{I}(\bar{x}) / \mathbb{V}(\bar{x})$.
    As $\mathbb{V}(x) = \mathbb{V}(\bar{x})$, $\mathbb{I}(x) \geq \mathbb{I}(\bar{x})$.

    Assume now the statement is correct up to some $k \geq 0$.
    Consider the case $k+1$, so $x_0, \dots, x_k = \bar{x}_0, \dots, \bar{x}_k$.
    If $x_i > 0$ for all $i \leq k$, by construction,\footnote{For ease of exposition, the proofs will henceforth restrict to the more challenging case of positive investments and we omit the easier (yet largely analogous) case that some agents invest zero. Besides being obviously inefficient (recall that $p$ increases steeply at zero), zero investments are easier to address as the relevant sums become finite and they require no new techniques. 
The results, nevertheless, should be understood as covering the full domain that includes zero investments.}
    \[
        \mathbb{V}(x)
        = 1 + p(x_0) + \dots + p(x_0) \cdots p(x_k) \mathbb{V}(x_{>k})
        = 1 + p(\bar{x}_0) + \dots + p(\bar{x}_0) \cdots p(\bar{x}_k) \mathbb{V}(\bar{x}_{>k})
        = \mathbb{V}(\bar{x}).
    \]
    Hence, $\mathbb{V}(x_{>k}) = \mathbb{V}(\bar{x}_{>k})$.
    We can now reapply the argument above but with respect to the profiles $x_{>k}$ and $\bar{x}_{>k}$ to conclude that $\mathbb{I}(x_{>k}) \geq \mathbb{I}(\bar{x}_{>k})$.
    Analogous to the above,
    \[
        \mathbb{I}(x)
        = x_0 + p(x_0) x_1 + \dots + p(x_0) \cdots p(x_k) \mathbb{I}(x_{>k})
        \geq \bar{x}_0 + p(\bar{x}_0) \bar{x}_1 + \dots + p(\bar{x}_0) \cdots p(\bar{x}_k) \mathbb{I}(\bar{x}_{>k})
        = \mathbb{I}(\bar{x}).
    \]
    By induction, the argument extends to any $k$.
\end{proof}

An immediate implication of Proposition~\ref{PR:flatten} is that the profile that maximizes welfare (without equilibrium constraints) is constant.
For such profiles $(c,c,\dots)$, we have
\[
    \mathbb{W}(c,c,\dots)
    = \frac{1 - c}{1 - p(c)},
\]
which is positive and single-peaked (Lemma~\ref{LE:math}, Appendix) for $c \leq 1$.
We then label its unique maximizer, the \textbf{first-best investment}, by $c^{FB} \in (0,1)$.
Maximizing the above with respect to $c$ yields that $c^{FB}$ is the unique $c \in (0,1)$ such that 
\[
    \mathbb{W}(c,c,\dots) = \frac{g(c)}{p(c)}.
\]

\begin{remark}[First-best investments cannot be supported] \label{REM:firstbest}
    As $p$ is concave, $g(c) > c$, for all $c>0$.
    In particular, for $c = c^{FB} > 0$, 
    \[
        \frac{g(c)}{p(c)} 
        = \frac{1 - c}{1 - p(c)} 
        > \frac{1 - g(c)}{1 - p(c)} \iff g(c) > p(c).
    \]
    It follows by Remark~\ref{REM:constant} that $(c^{FB},c^{FB},\dots)$ cannot be supported.
    $\hfill \circ$
\end{remark}

Next, Example~\ref{EX:overunder} makes two observations:
first, the initiator never invests more than $c^{FB}$ in equilibrium;
second, there are equilibria in which \emph{everyone else} invests more than $c^{FB}$.
The first part is complementary to Theorem~\ref{TH:equilibriumExistence}, which relied on individual, potentially fairly high bounds $B_i$ on equilibrium investments.
Example~\ref{EX:overunder} tightens the initiator's bound to $c^{FB}$.
On the other hand, the second part complements Example~\ref{EX:unboundedEq}.
To contrast, Example~\ref{EX:unboundedEq} identifies an equilibrium in which all but a finite number of agents invest more than $\mathcal{B}$, where the bound $\mathcal{B}$ may be arbitrarily high. 
Example~\ref{EX:overunder} considers the particular case $\mathcal{B} = c^{FB}$ and shows that, with a mild condition on $p$, there is an equilibrium in which all but the initiator invest more than $\mathcal{B}$.

\begin{example}[Equilibrium vs first-best investment] \label{EX:overunder}
    Given the two different observations, the example is split in two parts.
    First, we show that the initiator never invests more than $c^{FB}$ in equilibrium.

    \medskip
    \noindent \textsc{Part I:} \textit{Initiator bound.}
    Let $x$ be an equilibrium profile.
    By Lemma~\ref{OBS:aggregate}, $\mathbb{V}(x) - 1 \geq \mathbb{G}(x)$. 
    By construction, $g(x_i) \geq x_i$.
    Thus,
    \[
        \mathbb{G}(x)
        = g(x_0) + \sum_{j = 1}^{+\infty} \prod_{i = 0}^{j-1} p(x_i) g(x_j)
        \geq g(x_0) + \sum_{j = 1}^{+\infty} \prod_{i = 0}^{j-1} p(x_i) x_j
        = g(x_0) + p(x_0) \mathbb{I}(x_{>0}).
    \]
    Hence, $\mathbb{V}(x) - 1 = p(x_0) \mathbb{V}(x_{>0}) \geq g(x_0) + p(x_0) \mathbb{I}(x_{>0})$.
    Rearrange and use that $c^{FB}$ maximizes $\mathbb{W}$:
    \[
        \frac{g(x_0)}{p(x_0)} 
        \leq \mathbb{V}(x_{>0}) - \mathbb{I}(x_{>0}) 
        = \mathbb{W}(x_{>0})
        \leq \mathbb{W}(c^{FB}, c^{FB}, \dots)
        = \frac{g(c^{FB})}{p(c^{FB})}.
    \]
    As $g/p$ is increasing, we have $x_0 \leq c^{FB}$.

    \medskip
    \noindent \textsc{Part II:} \textit{Equilibrium overinvestment.}
    Let again $f$ be the rule $f^{JP}$ defined in Example~\ref{EX:rules} and $x$ be an equilibrium under $f$.
    Recall from Example~\ref{EX:unboundedEq} that, for each agent $i > 0$,
    \[
        \frac{g(x_i)}{p(x_i)} 
        = R_i(x_{>i},f) - f(i,i)
        = (1 - p(x_{i+1})) (i+1).
    \]
    Intuitively, if $p$ is ``close to zero'', then $g(x_i) / p(x_i) \approx i + 1$ whereas $g(c^{FB}) / p(c^{FB}) = (1 - c^{FB}) / (1 - p(c^{FB})) \approx 1$.
    As $g/p$ is increasing, we would then have that $x_i > c^{FB}$ for each agent $i > 0$.
    
    For a concrete example, set $\varepsilon = \sqrt{2} / 2$, so $p$ is bounded by $1 - \sqrt{2} / 2 \approx 0.3$.
    Then, for $i > 0$,
    \[
        (1 - p(x_{i+1})) (i+1)
        > 2 (1 - (1 - \sqrt{2}/2)) 
        = \sqrt{2}.
    \]
    The inequality is strict as $p$ is increasing.
    On the other hand,
    \[
        \frac{g(c^{FB})}{p(c^{FB})}
        = \frac{1 - c^{FB}}{1 - p(c^{FB})}
        \leq \frac{1}{1 - p(c^{FB})}
        < \frac{1}{1 - (1 - \sqrt{2} / 2)}
        = \sqrt{2}. 
    \]
    Hence, $x_i > c^{FB}$ for each agent $i > 0$.
    $\hfill \circ$
\end{example}

Although the first-best investment is a natural benchmark, our interest is primarily in \emph{equilibrium} profiles.
To that end, Remark~\ref{REM:firstbest} showed that the first-best investments cannot be supported.
Still, the conclusion of Proposition~\ref{PR:flatten}, that constant-tail profiles are cost effective, will remain very useful if we can show that flattening the tail of an \emph{equilibrium} profile results in a supportable profile.
Proposition~\ref{PR:cheaper} is a step in this direction:
it shows that constant-tail profiles reduce the ``cost of incentivizing'' the profile.
The proof parallels that of Proposition~\ref{PR:flatten}, but exploits that $g$ is convex rather than $p$ concave.
    
\begin{proposition}[Constant-tail profiles are cheaper to incentivize] \label{PR:cheaper}
    For each $x \in X$, $k \in \mathbb{N}_0$, and $\bar{x} = (x_0, \dots, x_{k-1}, c, c, \dots) \in X$ such that $\mathbb{V}(x) = \mathbb{V}(\bar{x})$, we have $\mathbb{G}(x) \geq \mathbb{G}(\bar{x})$.
\end{proposition}

\begin{proof}
    Say first $k = 0$, which corresponds to $\bar{x} = (c, c, \dots)$.
    By Proposition~\ref{PR:flatten}, $\mathbb{I}(x) \geq \mathbb{I}(\bar{x})$, where $\mathbb{I}(\bar{x}) = c \mathbb{V}(\bar{x}) = c \mathbb{V}(x)$.
    Hence, $c \leq \mathbb{I}(x) / \mathbb{V}(x)$.
    As $g$ is increasing, $g(c) \leq g \left ( \mathbb{I}(x) / \mathbb{V}(x) \right)$.
    As $g$ is convex,
    \begin{align*}
        g \left ( \frac{\mathbb{I}(x)}{\mathbb{V}(x)} \right) 
        &= g \left ( \frac{x_0 + p(x_0) x_1 + p(x_0) p(x_1) x_2 + \dots}{1 + p(x_0) + p(x_0) p(x_1) + \dots} \right) \\
        &\leq \frac{g(x_0) + p(x_0) g(x_1) + p(x_0) p(x_1) g(x_2) + \dots}{1 + p(x_0) + p(x_0) p(x_1) \dots}. 
    \end{align*}
    Hence, as $\mathbb{V}(x) = \mathbb{V}(\bar{x})$,
    \[
        g(c) 
        \leq g \left ( \frac{\mathbb{I}(x)}{\mathbb{V}(x)} \right )
        \leq \frac{\mathbb{G}(x)}{\mathbb{V}(\bar{x})}
        \implies
        \mathbb{G}(\bar{x}) = \mathbb{V}(\bar{x}) g(c) \leq \mathbb{G}(x).
    \]

    Assume now the statement is correct up to some $k \geq 0$.
    Consider the case $k+1$, so $x_0, \dots, x_k = \bar{x}_0, \dots, \bar{x}_k$.
    By construction,
    \[
        \mathbb{V}(x)
        = 1 + p(x_0) + \dots + p(x_0) \cdots p(x_k) \mathbb{V}(x_{>k})
        = 1 + p(\bar{x}_0) + \dots + p(\bar{x}_0) \cdots p(\bar{x}_k) \mathbb{V}(\bar{x}_{>k})
        = \mathbb{V}(\bar{x}).
    \]
    Hence, $\mathbb{V}(x_{>k}) = \mathbb{V}(\bar{x}_{>k})$.
    We can now reapply the argument above but with respect to the profiles $x_{>k}$ and $\bar{x}_{>k}$ to conclude that $\mathbb{G}(x_{>k}) \geq \mathbb{G}(\bar{x}_{>k})$.
    Analogous to the above,
    \begin{align*}
        \mathbb{G}(x)
        &= g(x_0) + p(x_0) g(x_1) + \dots + p(x_0) \cdots p(x_k) \mathbb{G}(x_{>k}) \\
        &\geq g(\bar{x}_0) + p(\bar{x}_0) g(\bar{x}_1) + \dots + p(\bar{x}_0) \cdots p(\bar{x}_k) \mathbb{G}(\bar{x}_{>k})
        = \mathbb{G}(\bar{x}).
    \end{align*}
    By induction, the argument extends to any $k$.
\end{proof}

The assumption that $g$ is convex is essential for Proposition~\ref{PR:cheaper} and the results that will follow.
If $g$ is not convex, then there can be non-constant equilibrium profiles for which the associated constant profile cannot be supported.
Example~\ref{EX:g} illustrates this point.

\begin{example}[The case of $g$ non-convex] \label{EX:g}
    Let $p$ be a smoothed version of the function illustrated in Figure~\ref{FIG:EX:g}.
    We first show that $x = (3/4, 1/4, 1/4, \dots)$ is supported by the rule $f$ as defined in Figure~\ref{FIG:EX:g}.
    
    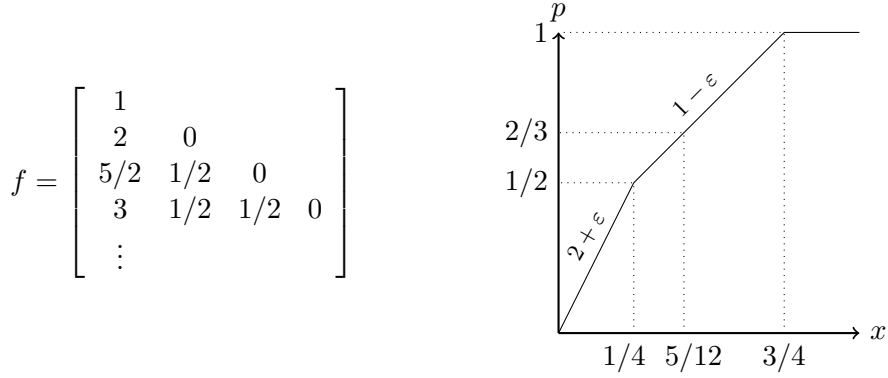
\begin{figure}[!bht]
        \centering
        \begin{tikzpicture}
            \node at (0,2) {
                \[
                    f = \left [ \begin{array}{cccc}
                            1 \\
                            2 & 0 \\
                            5/2 & 1/2 & 0 \\
                            3 & 1/2 & 1/2 & 0 \\
                            \vdots
                    \end{array} \right ]
                \]
            };
            \draw [thick,->] (0,0) -- (4,0) node [right] {$x$};
            \draw [thick,->] (0,0) -- (0,4) node [above] {$p$};
            \draw (0,0) -- (1,2) node [rotate=60,midway,above] {\quad \footnotesize $2+\varepsilon$} -- (3,4) node [rotate=45,midway,above] {\footnotesize $1 - \varepsilon$} -- (4,4);
            \draw [dotted] (1,0) node [below] {1/4 \ \ } -- (1,2) -- (0,2) node [left] {1/2};
            \draw [dotted] (3,0) node [below] {3/4} -- (3,4) -- (0,4) node [left] {$1$};
            \draw [dotted] (5/3,0) node [below] {\ \ 5/12} -- (5/3,8/3) -- (0,8/3) node [left] {2/3};
        \end{tikzpicture}
        \caption{Success rate $p$ with approximate slopes indicated above the lines.}
        \label{FIG:EX:g}
    \end{figure}
    
    We first verify for agents $i > 0$.
    By Lemma~\ref{OBS:individual}, we should have $g(x_i) / p(x_i) = R_i(x_{>i}, f) - f(i,i) = 1/2$.
    By definition, $g(x_i) / p(x_i) = 1 / p'(x_i)$.
    As it is possible to smoothen $p$ such that $p'(1/4) = 2$, we obtain that $x_i = 1/4$ indeed is optimal.
    Note also that $p(x_1) = p(x_2) = \dots = 1/2$.
    
    For agent $0$, on the other hand,
    \begin{align*}
        R_0(x_{>0}, f) - f(0,0)
        = 1 + p(x_1) \cdot (1/2) + p(x_1) p(x_2) \cdot (1/2) + \dots
        = 1 + 1/4 + 1/8 + \dots
        = 3/2.
    \end{align*}
    Again, we can smooth $p$ such that $p'(3/4) = 2/3$.
    Hence, we conclude that $f$ supports $x$.
    Moreover, 
    \[
        \mathbb{V}(x) 
        = 1 + p(x_0) + p(x_0) p(x_1) + \dots
        = 1 + 1 + 1/2 + 1/4 + \dots
        = 3.
    \]
    
    The corresponding constant profile $\bar{x} = (c, c, \dots)$ is such that
    \[
        \mathbb{V}(c,c,\dots) = \frac{1}{1 - p(c)} = 3.
    \]
    That is, $p(c) = 2/3$, so $c = 5/12$.
    Again, we can smooth $p$ such that $p'(c) < 1$, which means $p(c) < g(c)$.
    By Remark~\ref{REM:constant}, such a constant profile $\bar{x}$ cannot be supported.
    $\hfill \circ$
\end{example}

\section{Optimality} \label{SEC:optimality}

In this section, we focus on results that characterize the optimal investment profiles that can be supported in equilibrium. In particular, we look for welfare maximizing and initiator maximizing equilibrium profiles. The former objective is natural if we imagine the rule being designed by a (benevolent) social planner, while the latter is natural if designed by a (payoff-maximizing) initiator.

\subsection{Socially optimal equilibrium}
We already know, from the structural results in Section 3.2, that flattening investment profiles is welfare improving. As already noted, the first-best investment profile, $c^{FB}$, is constant. Yet, $c^{FB}$ cannot be supported in equilibrium (see Remark 2). The question is therefore whether there exists an equilibrium profile $x^*$ that maximizes $\mathbb{W}$ among all equilibrium profiles. Theorem \ref{TH:socialopt} below answers this affirmatively and further shows that such a profile is unique and can be supported by the canonical equal split rule, $f^{ES}$.

\begin{theorem} \label{TH:socialopt}
There exists a unique welfare-maximizing equilibrium profile. It is the constant profile $x^* = (c^*, c^*, \dots)$, where $c^* > 0$ is such that $g(c^*) = p(c^*)$. Moreover, $x^*$ is supported by $f^{ES}$.
\end{theorem}

\begin{proof}
    Let $x \in X$ be an equilibrium and let $\bar{x} = (c,c,\dots) \in X$ be such that $\mathbb{V}(\bar{x}) = \mathbb{V}(x)$.
    By Proposition~\ref{PR:flatten}, $\mathbb{I}(x) \geq \mathbb{I}(\bar{x})$ and, thus, $\mathbb{W}(x) \leq \mathbb{W}(\bar{x})$.
    By Lemma~\ref{OBS:aggregate}, $\mathbb{V}(x) - 1 \geq \mathbb{G}(x)$.
    By Proposition~\ref{PR:cheaper}, $\mathbb{G}(x) \geq \mathbb{G}(\bar{x})$.
    Hence, $\mathbb{V}(\bar{x}) - 1 \geq \mathbb{G}(\bar{x})$, or, equivalently, 
    \[
        \frac{p(c)}{ (1 - p(c))} \geq \frac{g(c)}{ (1 - p(c))}.
    \]
    That is, $p(c) \geq g(c)$.
    Thus, by Remark~\ref{REM:constant}, $\bar{x}$ can also be supported.
    Hence, for the purpose of maximizing welfare, it suffices to check constant profiles.
    For such profiles, 
    \[
        \mathbb{W}(c,c,\dots) = \frac{1 - c} {1 - p(c)}.
    \]
    Now, such a function is single-peaked for $c \in [0,1]$ and maximized at $c^{FB} > c^*$, where $c^*$ is such that $g(c^*) = p(c^*)$.\footnote{See Lemma~\ref{LE:math} in the Appendix for a formal proof.}
    Therefore, $\mathbb{W}(c,c,\dots)$ is maximized at $c^*$ for $c \in [0,c^*]$. Finally, by Remark~\ref{REM:constant}, the equal split rule $f^{ES}$ induces $x^* = (c^*, c^*, \dots)$ as unique equilibrium.
\end{proof}

 Even though $x^*$ is uniquely determined, $f^{ES}$ is not unique in supporting $x^*$, as shown below. 

\begin{example}[Multiplicity of optimal rules] \label{EX:mult-optrules}
For each $\left\vert\beta\right\vert \leq 1$, let $f^\beta \in F$ be defined as $f^{ES}$ everywhere, except for the following:
\begin{align*}
    f^\beta(0,2) &= 2 - \beta p(c),  
    f^\beta(0,3) = f^\beta(0,4) = \dots = 2 + \beta (1 - p(c)) \\
    f^\beta(1,2) &= 1 + \beta p(c),
    f^\beta(1,3) = f^\beta(1,4) = \dots = 1 - \beta (1 - p(c)) .
\end{align*}
With $\beta = 0$, we obtain $f^\beta = f^{ES}$.
As $\left\vert\beta\right\vert \leq 1$, $f^\beta(i,j) \geq 0$ for all $i$ and $j$.
Moreover, as only agents $0$ and $1$ are affected compared to $f^{ES}$ and $f^\beta(0,j) + f^\beta(1,j) = f^{ES}(0,j) + f^{ES}(1,j)$, $f^\beta$ is well-defined.
For each of exposition, let $p \equiv p(c)$. Then,
\begin{align*}
    R_0(x_{>0}^*, f^\beta)
    &= (1 - p) f(0,1) + p (1 - p) f(0,2) + p^2 (1 - p) f(0,3) + \dots \\
    &= (1 - p) \cdot 2 + p (1 - p) (2 - \beta p) + p^2 (1 - p) (1 + p + p^2 + \dots) (2 + \beta (1 - p)) \\
    &= 2 - 2p + 2 p - 2p^2 - \beta p^2 (1 - p) + 2 p^2 + \beta p^2 (1 - p) = 2 = R_0(x_{>0}^*,f^{ES}).
\end{align*}
Similarly,
\begin{align*}
    R_1(x^*,f^\beta)
    &= (1 - p) f(1,2) + p (1 - p) f(1,3) + p^2 (1 - p) f(1,4) + \dots \\
    &= (1 - p) (1 + \beta p) + p (1 - p) (1 + p + p^2 + \dots) (1 - \beta (1 - p) ) \\
    &= 1 - p + \beta p (1 - p) + p - \beta p (1 - p) = 1 
    = R_1(x^*,f^{ES}).
\end{align*}
The other agents are unaffected by the change.
Therefore, for each agent $i$, we have $R_i(x_{>i}^*,f^\beta) = R_i(x_{>i}^*,f^{ES})$.
As both $x^*$ and $f(i,i)$, for each $i$, are unchanged, we have $U_i(x^*,f^\beta) = U_i(x^*,f^{ES}) = \hat{U}_i(x^*_i,f^{ES}) = \hat{U}_i(x^*_i,f^\beta)$. Thus, $(x^*, f^\beta)$ is an equilibrium. \hfill $\circ$
\end{example}

In spite of what Example \ref{EX:mult-optrules} shows, there are other aspects in which the equal split rule $f^{ES}$ is unique. First, it the only rule that supports $x^*$ in equilibrium for which no information about the success rate $p(\cdot)$ is used. This seems particularly convenient for practical implementation because, by choosing $f^{ES}$ to implement the socially optimal profile $x^*$, an external benevolent planner can be completely unaware of the actual success rate and still be sure to induce the right equilibrium (as long as $p(\cdot)$ is known by the agents). Moreover, $f^{ES}$ is also the only rule that supports $x^*$ for which no agent experiences a decrease in payoff at some point in the process. 
To show this, take any rule $f \in F$.
Then,
\begin{align*}
    R_0(x_{>0}^*,f)
    &= (1 - p) f(0,1) + p (1 - p) f(0,2) + p^2 (1 - p) f(0,3) + \dots \\
    &\geq (1 + p + p^2 + \dots) (1 - p) f(0,1) = f(0,1) = 2,
\end{align*}
where the inequality is strict if $f(0,i) > f(0,1)$ for some $i > 1$.
Once we have pinned down agent $0$'s payoffs, we can turn to agent $1$.
We have
\begin{align*}
    R_1(x^*,f)
    &= (1 - p) f(1,2) + p (1 - p) f(1,3) + p^2 (1 - p) f(1,4) + \dots \\
    &\geq (1 + p + p^2 + \dots) (1 - p) f(1,2) = f(1,2) = 3 - f(0,2) = 1,
\end{align*}
where, again, the inequality is strict if $f(1,i) > f(1,2)$ for some $i > 2$.
Proceeding this way, we find that $f = f^{ES}$.

Finally, we note that common to all the rules that support $x^*$ is that $f(i,i) = 0$ for all $i > 0$. Intuitively, rules for which $f(i,i)>0$, for some $i>0$, reduce the ``budget" for incentivizing agents: in this sense, $f(i,i)$ can be viewed as ``dead capital" that could have been used to incentivize the creation of more value. 

\subsection{Initiator optimal equilibrium} \label{SUB:characterization}
We now turn to the other optimality criterion we consider, and ask whether there exists an equilibrium profile that maximizes the initiator's expected payoff $U_0$, among all equilibrium profiles. Imagine, for instance, that the initiator is allowed to choose the rule ex ante to her own advantage, setting the process in motion herself. Intuitively, one would expect that the initiator would choose a rule for which the expected value of the subgame starting from agent 1 is maximized, and indeed, it turns out that initiator optimal profiles will have the form of near-constant profiles, i.e., $(x_0, c, c, \dots)$ where all agents invest the same, except possibly the initiator herself.

We therefore need to examine the conditions under which near-constant profiles can be supported. Proposition~\ref{PR:nearconstant} does exactly that. Notice, up front, that the condition is satisfied, for instance, when $x_0 \geq c$.

\begin{proposition}[Flattening to a near-constant profile can maintain equilibrium] \label{PR:nearconstant}
    Let $x \in X$ with $x_0 > 0$ be an equilibrium profile and define $c \geq 0$ and $\bar{x} = (x_0, c, c, \dots) \in X$ such that $\mathbb{V}(\bar{x}) = \mathbb{V}(x)$.
    Then, $\bar{x}$ can be supported if and only if
    \[
        \frac{g(c)}{p(c)} - 2 \leq \frac{g(x_0)}{p(x_0)}.
    \]
\end{proposition}

\begin{proof}
    Let $x$ be an equilibrium profile with $x_0 > 0$.
    By Lemma~\ref{LE:characterization} in the Appendix, a generic near-constant profile $(x_0, c, c, \dots)$ can be supported if and only if
    \[
        \frac{g(c)}{p(c)} - 2 \leq \frac{g(x_0)}{p(x_0)} \leq \frac{1 - g(c)}{1 - p(c)}.
    \]
    Define now $\bar{x} = (x_0, c, c, \dots)$ with $c \geq 0$ such that $\mathbb{V}(\bar{x}) = \mathbb{V}(x)$.
    To complete the proof, it remains to show that the upper bound is satisfied.
    That is, if $\bar{x} = (x_0, c, c, \dots)$ is created on the basis of an equilibrium profile $x$, then $$\frac{g(x_0) } {p(x_0)} \leq \frac{1 - g(c)} {1 - p(c)}.$$
    
    To show that, note that, as $x$ is an equilibrium profile, it follows from  Lemma~\ref{OBS:aggregate} that $\mathbb{V}(x) - 1 \geq \mathbb{G}(x)$.
    By Proposition~\ref{PR:cheaper}, $\mathbb{G}(x) \geq \mathbb{G}(\bar{x})$.
    Hence, 
    \[
        \mathbb{G}(\bar{x}) 
        = g(x_0) + \frac{p(x_0) g(c)}{1 - p(c)}
        \leq \mathbb{G}(x)
        \leq \mathbb{V}(x) - 1
        = \mathbb{V}(\bar{x}) - 1
        = \frac{p(x_0)}{1 - p(c)}.
    \]
    Thus,
    \[
        g(x_0) \leq \frac{p(x_0)(1- g(c))}{1 - p(c)},
    \]
    as desired.
\end{proof}

If the inequality in the statement of Proposition~\ref{PR:nearconstant} 
is not satisfied, then flattening an equilibrium profile results in a profile that cannot be supported.
Example~\ref{EX:nearconstantNotEq} illustrates this point.

\begin{example}[Flattened equilibrium profile cannot be supported] \label{EX:nearconstantNotEq}
    Let the profile $x$ be supported by the rule $f$ defined below.
    As $f$ is symmetric for agents $2, 3, \dots$, we have $x_2 = x_3 = \dots$, which we use in the equilibrium condition for $i \geq 2$ below.
    \[
        f = \left[ \begin{array}{cccccc} 
            1 \\
            2 & 0 \\
            0 & 3 & 0 \\
            0 & 2 & 2 & 0 \\
            0 & 2 & 1 & 2 & 0 \\
            0 & 2 & 1 & 1 & 2 & 0 \\
            \vdots
        \end{array} \right] 
        \implies
        \frac{g(x_i)}{p(x_i)} =
        R_i(x_{>i},f) - f(i,i) =
        \begin{cases}
            1 - 2 p(x_1) & \text{for } i = 0 \\
            3 - p(x_2) & \text{for } i = 1 \\
            2 - p(x_2) & \text{for } i \geq 2 
        \end{cases}
    \]
    Let the success rate be given by%
    \footnote{Strictly speaking, $p$ is not bounded by $1 - \varepsilon$ here.
    The function chosen here is relatively easy to work with, but the same intuition would hold for instance for $\tilde{p}(x) = (1 - \varepsilon) p(x)$.}
    \[
        p(x) = \frac{\sqrt{x}}{1 + \sqrt{x}}.
    \]
    This yields $x_2 \approx .1777$, $x_1 \approx .3106$, and $x_0 \approx .0131$.
    The near-constant profile $\bar{x} = (x_0, c, c, \dots)$ with $\mathbb{V}(\bar{x}) = \mathbb{V}(x)$ has $c \approx .2588$.
    Then $g(x_0) / p(x_0) \approx .2842 < .3160 \approx g(c) / p(c) - 2$.
    By Proposition~\ref{PR:nearconstant}, $\bar{x}$ cannot be supported. \hfill $\circ$
    \end{example}

We are now ready to state our third main result.

\begin{theorem} \label{TH:initiator}
There exists a unique initiator optimal equilibrium profile. It is the near-constant profile $x^\circ = (x_0^\circ, c^\circ, c^\circ,\dots)$, where $c^\circ$ and $x_0^\circ$ are such that $g'(c^\circ)(1-p(c^\circ)) = p'(c^\circ)(1-g(c^\circ))$ and $g(x_0^\circ)(1-p(c^\circ))= p(x_0^\circ)(1-g(c^\circ))$. Moreover, $x^{\circ}$ is supported by $f^\alpha$, with $\alpha = \frac{g(c^\circ)}{p(c^\circ)}.$
\end{theorem}

\begin{proof}
    Let $x$ be an equilibrium profile. Then, the initiator's payoff is $1 + g(x_0) - x_0$.
    By definition, $g(0) = 0$.
    By construction, $g(x_0) > x_0$ for $x_0 > 0$.
    As $g$ is convex, it follows that, among equilibrium profiles, the initiator's payoff is increasing in $x_0$.
    Next, we argue that, for the purpose of maximizing the initiator's equilibrium payoff, it suffices to consider near-constant profiles $(x_0, c, c, \dots)$ with $x_0 \geq c$.

    Let $x$ be an arbitrary equilibrium profile and define $\bar{x} = (x_0, c, c, \dots)$ with $\mathbb{V}(\bar{x}) = \mathbb{V}(x)$. Moreover, let $y$ be a constant profile such that 
     $\mathbb{V}(\bar x) = \mathbb{V}(y)$.
    If $x_0 \geq c$, then, by Proposition~\ref{PR:nearconstant}, $\bar x$ can be supported. 
    If instead $x_0 < c$, then $y > x_0$, and by Proposition~\ref{PR:nearconstant}, $y$ can be supported.
    In this way, the initiator's payoff is at least as high at the equilibrium profiles $\bar x$ and $y$ as at $x$.
    Hence, for the purpose of maximizing the initiator's equilibrium payoff, it suffices to consider profiles of the form $(x_0,c,c,\dots)$ with $x_0 \geq c$.
    
    By Lemma~\ref{LE:characterization} in the appendix, a profile $(x_0,c,c,\dots)$ with $x_0 \geq c$ can be supported in equilibrium if and only if
    \[
        \frac{1}{p'(x_0)} = \frac{g(x_0)}{p(x_0)} \leq \frac{1 - g(c)}{1 - p(c)}.
    \]
    As $p$ is concave in $x_0$, $p'$ is decreasing and $1/p'$ is increasing in $x_0$.
    Therefore, we must have equality at an initiator-optimal profile, as otherwise we could increase $x_0$ and make the initiator better off.
    Similarly, the right-hand side must be maximized with respect to $c$, as otherwise we could find a different ``continuation'' $c'$ to support $x_0' > x_0$ such that $(x_0',c',c',\dots)$ could be made an equilibrium under which the initiator is better off.
    
    To conclude, let $d$ be such that $g(d) = 1$.
    By Lemma~\ref{LE:math} in the appendix, $\frac{1 - g(c)} {1 - p(c)}$ is single-peaked on $[0,d]$.
    The maximizer $c$ then satisfies 
    \[
        p'(c) (1 - g(c)) = g'(c) (1 - p(c)),
    \]
    which is precisely the condition defining $c^\circ$.
    The level $x_0^\circ$ then follows from
    \[
        \frac{g(x^\circ_0)}{p(x^\circ_0)} = \frac{1 - g(c^\circ)}{1 - p(c^\circ)}. \qedhere
    \]
\end{proof}

As in the case of the socially optimal equilibrium profile, the initiator optimal profile $x^\circ$ can be supported by multiple rules: even multiple rules for which no agent experiences a decrease in payoff at some point in the process \cite[see also][]{hougaard2022optimal}. Common to all the supporting rules is their dependence of the success rate. Yet, as $p$ is common knowledge among the agents, the initiator can easily construct the appropriate rule for implementation.  

\bigskip
The example below illustrates the set of potential near-constant equilibrium profiles as well as compare the socially, and initiator optimal profiles. 

\begin{example} (Illustration of Theorems 2 and 3). Consider the specific success rate 
\[
    p(x) = \frac{\sqrt{x}}{1+\sqrt{x}}.
\]
Figure 2 illustrates the set of near-constant equilibrium profiles as those under the red curve and above the blue curve (in particular, notice that $\bar x$, from Example \ref{EX:nearconstantNotEq} above, is below the blue curve and hence cannot be supported). The red curve corresponds to the constraint $\mathbb{V}(x) - 1 \geq \mathbb{G}(x)$. The blue curve corresponds to the constraint given in Proposition \ref{PR:nearconstant}: The intersection between the red and blue curve is at the value $c$ such that $\frac{g(c)}{p(c)} = 3 - 2 p(c)$.

The socially optimal profile $c^*$ is determined by the tangent between the red curve and the highest level curve of $\mathbb{W}(x_0,c,c,\dots)$. In particular, we get $x_0=c=c^*$. As the initiator's payoff is increasing in her own investment, $x_0$, the initiator optimal profile $x^\circ$ is determined by the value of $c^\circ$ that maximizes $x_0$ on the red curve. Clearly, $x_0^\circ > c^\circ$.

 \begin{figure}[!htb]
        \centering
        \begin{tikzpicture}[] %
            \begin{axis}[
                width=.8\textwidth,
                height=.36\textwidth,
                xmin = 0,
                xmax = 0.3,
                ymin = 0,
                ymax = 0.12,
                ticks = none,
                axis x line = bottom,
                axis y line = left,
                xlabel = {$c$},
                ylabel = {$x_0$},
                style={thick},
                every axis x label/.style={at={(current axis.right of origin)},anchor=west},
                every axis y label/.style={at={(current axis.north west)},above},
                ]
                \node [fill, black, circle, inner sep = 1.5pt] at (axis cs: .2588, .0131) {};
                \node [fill, white, circle, inner sep = 1pt] at (axis cs: .2588, .0131) {};
                \node [fill, black, circle, inner sep = .5pt] at (axis cs: .2588, .0131) {};
                \node [below] at (axis cs: .2588, .0131) {$\bar{x}$};
                \draw [dotted] (axis cs: .088, .088) -- (axis cs: .088, 0) node [xshift=.7em, yshift=.7em] {$c^*$};
                \draw [dotted] (axis cs: .026, .099) -- (axis cs: .026, 0) node [xshift=.7em, yshift=.7em] {$c^\circ$};
                \addplot[black!50, dashed] table [skip first n=1, x index = {1}, y index = {1}, col sep=comma] {Data/data1.csv};
                \addplot[red] table [skip first n=1, x index = {1}, y index = {2}, col sep=comma] {Data/data1.csv};
                \addplot[blue] table [skip first n=1, x index = {1}, y index = {3}, col sep=comma] {Data/data1.csv};
                \addplot[black!50,thin] table [skip first n=70, x index = {1}, y index = {2}, col sep=comma] {Data/data3.csv};
            \end{axis}
        \end{tikzpicture}
        \caption{The set of equilibrium near-constant profiles $(x_0, c, c, \dots)$ are above the blue and below the red curve.
        The success rate is $p(x) =  \frac{\sqrt{x}}{1+\sqrt{x}}$.}
        \label{FIG:counterexample}
    \end{figure}
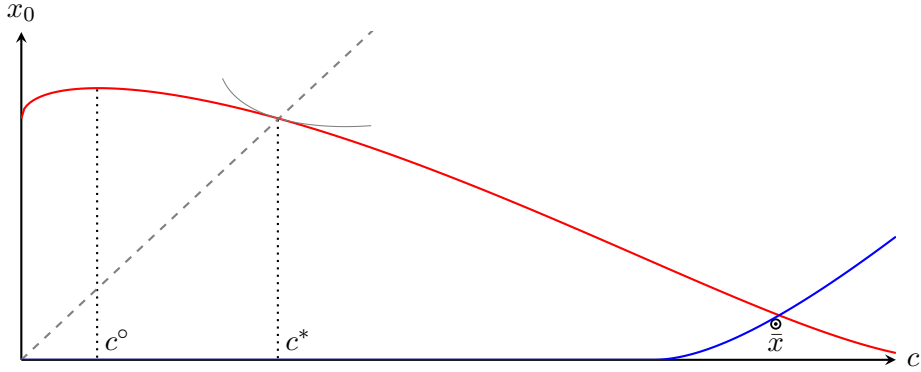

\end{example}

\section{Optimal rules with endogenous budget constraints} \label{SEC:budget}

In the baseline model studied so far, agents are not budget constrained and can make arbitrarily large investments.
Indeed, as Example~\ref{EX:unboundedEq} showed, even equilibrium investments might not be bounded. Yet, there may be situations where agents have limited access to funds. In this section, we take this to its extreme and assume that the only funds that an agent has available for investment are the ones she has been awarded through the value-creating process. We model this through the payments $f(i,i)$, so positive values of $f(i,i)$ now become essential as a budget limitation. That is, $x_i \leq f(i,i)$. And, as these funds are set aside for agent $i$ throughout the process, in effect, we restrict attention to rules such that $f(i,i) \leq f(i,j)$ for all $i$ and $j>i$.
 
Now, the concepts defined in Section~\ref{SEC:model} naturally extend to this scenario.
The profile $x$ is a {\bf self-financed} equilibrium profile in the game induced by the rule $f\in F$ if, for each agent $i$, the investment $x_i\leq f(i,i) \leq f(i,j)$ maximizes $U_i(x_{\geq i}, f)$ given the investment $x_{\geq i}$ of those who succeed $i$ under $f$.

Revisiting Lemmas~\ref{OBS:individual} and~\ref{OBS:aggregate} for the case of self-financed equilibria, we first notice that now we may have an equilibrium with corner solution $x_i = f(i,i)$ such that the first order condition is not satisfied. In this case, $i$ is ``under-investing'' given her expected payoff. That is, $R_i(x_{>i}, f) - f(i,i) \geq \frac{g(x_i)} {p(x_i)}$. Thus, in a self-financed equilibrium, $\mathbb{V}(x) \geq \hat{\mathbb{V}}(x,f)$.
Moreover, as $f(i,i) \geq x_i$, we now have
\[
    \mathbb{V}(x)
    \geq \hat{\mathbb{V}}(x,f)
    \geq 1 + p(x_0) \mathbb{I}(x_{>0}) + \mathbb{G}(x)
    \geq \mathbb{I}(x) + \mathbb{G}(x).
\]

In terms of socially optimal equilibria for our baseline model, Theorem~\ref{TH:socialopt} showed that supporting rules are designed to make all agents face identical investment decisions.
But in the present case of budget constraints, there is a significant difference between the initiator and the successive agents.
Regardless of the rule, the initiator can invest up to $f(0,0) = 1$ by construction, while all other agents will generally face a tighter budget constraint $f(i,i) < 1$.
As such, we see that the socially optimal self-financed equilibrium has a more similar structure to the initiator-optimal equilibrium without budget constraints from Theorem~\ref{TH:initiator} (which was a near-constant profile where the initiator invests more). 

\begin{theorem}\label{TH:SF-eq}
There exists a welfare-maximizing self-financed equilibrium profile. It is a near-constant profile $x^s= (x_0^s,c^s,c^s,\dots)$ determined by the program
\[
    \max_{x_0,c} \mathbb{W}(x_0,c,c\dots) = 1-x_0+\frac{p(x_0)(1-c)}{1-p(c)} \ \  \mbox{such that} \ \ \frac{g(x_0)}{p(x_0)}= \frac{1-c-g(c)}{1-p(c)},
\]
Moreover, $x^s$ is supported by the rule $f^{\alpha,\F}$ with $\alpha = \frac{g(c^s)}{p(c^s)}+c^s$ and $\F = c^s.$
\end{theorem}

\begin{proof}
    We start showing that the search for welfare-maximizing self-financed equilibria can be restricted to near-constant and constant profiles. Let $x$ be a self-financed equilibrium under $f$. Let $c \geq 0$ and $\bar{x} = (x_0, c, c, \dots)$ be such that $\mathbb{V}(\bar{x}) = \mathbb{V}(x)$. We now distinguish two cases, depending on whether the initiator's investment is above or below the (constant) investment all successors make (in the near-constant profile). 
\medskip

    \textsc{Part 1:}
    Suppose that $x_0 \geq c$.
    Then,
    \[
        \mathbb{V}(\bar{x})
        = \mathbb{V}(x)
        \geq \hat{\mathbb{V}}(x,f)
        \geq 1 + p(x_0) \mathbb{I}(x_{>0}) + \mathbb{G}(x)
        \geq 1 + p(x_0) \mathbb{I}(\bar{x}_{>0}) + \mathbb{G}(\bar{x}).
    \]
    Equivalently,
    \[
        1 + \frac{p(x_0)}{1 - p(c)}
        \geq 1 + p(x_0) \cdot \frac{c}{1 - p(c)} + g(x_0) + \frac{p(x_0) g(c)}{1 - p(c)}.
    \]
    Hence,
    \[
        \frac{g(x_0)}{p(x_0)} \leq \frac{1 - g(c) - c}{1 - p(c)}.
    \]
    By Lemma~\ref{LE:characterization} in the appendix, $\bar{x}$ can also be supported in a self-financed equilibrium.
    
    By Proposition~\ref{PR:flatten}, $\mathbb{I}(\bar{x}) \leq \mathbb{I}(x)$ and, thus, $\mathbb{W}(\bar{x}) \geq \mathbb{W}(x)$.
     \medskip
     
    \textsc{Part 2:}
    Suppose that $x_0 < c$.
    Define now $\tilde{c} \geq 0$ and redefine $\bar{x} \equiv (\tilde{c},\tilde{c},\dots)$ such that, again, $\mathbb{V}(\bar{x}) = \mathbb{V}(x)$ (that is, flatten the profile fully now).
    As $\mathbb{V}$ is monotonic, we have $x_0 < \tilde{c} < c$.
    By Proposition~\ref{PR:flatten}, $\mathbb{I}(\bar{x}) \leq \mathbb{I}(x)$.
    That is,
    \[
        \tilde{c} + p(\tilde{c}) \mathbb{I}(\bar{x}_{>0})
        \leq x_0 + p(x_0) \mathbb{I}(x_{>0}).
    \]
    As $x_0 < \tilde{c}$,
    \[
        p(\tilde{c}) \mathbb{I}(\bar{x}_{>0}) < p(x_0) \mathbb{I}(x_{>0}).
    \]
    It follows now that
    \[
        \mathbb{V}(\bar{x})
        = \mathbb{V}(x)
        \geq \hat{\mathbb{V}}(x,f)
        \geq 1 + p(x_0) \mathbb{I}(x_{>0}) + \mathbb{G}(x)
        \geq 1 + p(\tilde{c}) \mathbb{I}(\bar{x}_{>0}) + \mathbb{G}(\bar{x}).
    \]
    Equivalently,
    \[
        \frac{1}{1 - p(\tilde{c})}
        \geq 1 + p(\tilde{c}) \cdot \frac{\tilde{c}}{1 - p(\tilde{c})} + \frac{g(\tilde{c})}{1 - p(\tilde{c})}.
    \]
    Equivalently, 
    \[
        1 \geq 1 - p(\tilde{c}) + p(\tilde{c}) \tilde{c} + g(\tilde{c}).
    \]
    Rearranging,
    \[
        \frac{g(\tilde{c})}{p(\tilde{c})} \leq 1 - \tilde{c}.
    \]
    By Lemma~\ref{LE:characterization} in the appendix, $\bar{x}$ can also be supported in a self-financing equilibrium.
    By Proposition~\ref{PR:flatten}, $\mathbb{I}(x) \geq \mathbb{I}(\bar{x})$ and, thus, $\mathbb{W}(x) \leq \mathbb{W}(\bar{x})$.

    \textsc{Part 3:}
    From Parts~1 and~2, we conclude that, for any self-financed equilibrium profile $x$, the near-constant $\bar{x} = (x_0, c, c, \dots)$ with $\mathbb{V}(\bar{x}) = \mathbb{V}(x)$ can also be supported in self-financed equilibrium and is such that $\mathbb{W}(\bar{x}) \geq \mathbb{W}(x)$.
    That is to say, for the purpose of maximizing $\mathbb{W}$, it suffices to restrict to near-constant profiles.
    For $\bar{x} = (x_0, c, c, \dots)$, we have
    \[
        \mathbb{W}(\bar{x}) = 1 - x_0 + \frac{p(x_0) (1 - c)}{1 - p(c)},
    \]
    which is precisely the objective function of the program.
    By Lemma~\ref{LE:characterization} with $\gamma = c$, $\bar{x} = (x_0, c, c, \dots)$ can be supported if and only if
    \[
        \frac{g(c)}{p(c)} + c - 2 \leq \frac{g(x_0)}{p(x_0)} \leq \frac{1 - g(c) - c}{1 - p(c)}.
    \]
    As $g(c) \geq 0$, we have
    \[
        \frac{1 - c}{1 - p(c)} \geq \frac{1 - g(c) - c}{1 - p(c)} \geq \frac{g(x_0)}{p(x_0)}.
    \]
    Now, differentiate $\mathbb{W}(\bar{x})$ with respect to $x_0$:
    \[
        p'(x_0) \cdot \frac{1 - c}{1 - p(c)} - 1
        = p'(x_0) \cdot \left ( \frac{1 - c}{1 - p(c)} - \frac{1}{p'(x_0)} \right )
        = p'(x_0) \cdot \left ( \frac{1 - c}{1 - p(c)} - \frac{g(x_0)}{p(x_0)} \right )
        \geq 0.
    \]
    Hence, $\mathbb{W}(\bar{x})$ is non-decreasing in $x_0$.
    As $g(x_0)/p(x_0)$ is increasing in $x_0$, it suffices to restrict to profiles of the form $\bar{x} = (x_0, c, c, \dots)$ such that $g(x_0) / p(x_0) = (1 - g(c) - c) / (1 - p(c))$.
    This is precisely the constraint of the program.
\end{proof}

 The example below summarizes and compares optimal equilibrium profiles with, and without, budget constraints. 

\begin{example}(Illustration of Theorem 4)
   Consider (again) the specific success rate 
   \[
        p(x) = \frac{\sqrt{x}}{1+\sqrt{x}}.
    \]
    Compared to Figure 2, the corresponding dashed blue and red curves resulting from adding the budget constraint $x_i \leq f(i,i)$ are shown in Figure 3. It is clear that both curves become ``steeper" by adding the constraint, so the set of potential near-constant equilibria is effectively reduced. In particular, we see that the social optimum is changed as $c^*$ can no longer be supported. We now get that welfare is optimized at $x_0^s \approx .0815$ and $c^s \approx .0724$. So now, in the social optimum, the initiator will invest strictly more than the other agents, as illustrated by adding the level curves of $\mathbb{W}(x_0,c,c,\dots)$. Thus, note that, in both cases we have $x_0\geq c$: it is simply inefficient for the initiator to invest less than the other agents, even though such profiles can be supported in equilibrium.

    \begin{figure}[!htb]
        \centering
        \begin{tikzpicture}[] %
            \begin{axis}[
                width=.8\textwidth,
                height=.36\textwidth,
                xmin = 0,
                xmax = 0.3,
                ymin = 0,
                ymax = 0.12,
                ticks = none,
                axis x line = bottom,
                axis y line = left,
                xlabel = {$c$},
                ylabel = {$x_0$},
                style={thick},
                every axis x label/.style={at={(current axis.right of origin)},anchor=west},
                every axis y label/.style={at={(current axis.north west)},above},
                ]
                \draw [dotted] (axis cs: .088, .088) -- (axis cs: .088, 0) node [xshift=.7em, yshift=.7em] {$c^*$};
                \draw [dotted] (axis cs: .0724, .0815) -- (axis cs: .0724, 0) node [xshift=.7em, yshift=.7em] {$c^s$};
                \addplot[black!50, dashed] table [skip first n=1, x index = {1}, y index = {1}, col sep=comma] {Data/data1.csv};
                \addplot[red] table [skip first n=1, x index = {1}, y index = {2}, col sep=comma] {Data/data1.csv};
                \addplot[blue] table [skip first n=1, x index = {1}, y index = {3}, col sep=comma] {Data/data1.csv};
                \addplot[red,dashed] table [skip first n=1, x index = {1}, y index = {2}, col sep=comma] {Data/data2.csv};
                \addplot[blue,dashed] table [skip first n=1, x index = {1}, y index = {3}, col sep=comma] {Data/data2.csv};
                \addplot[black!50,thin] table [skip first n=70, x index = {1}, y index = {2}, col sep=comma] {Data/data3.csv};
                \addplot[black!50,thin] table [skip first n=55, x index = {1}, y index = {3}, col sep=comma] {Data/data3.csv};
            \end{axis}
        \end{tikzpicture}
        \caption{The set of equilibrium near-constant profiles $(x_0, c, c, \dots)$ with budget constraints. Comparing socially optimal equilibrium profiles.}
        \label{FIG:overview}
    \end{figure}
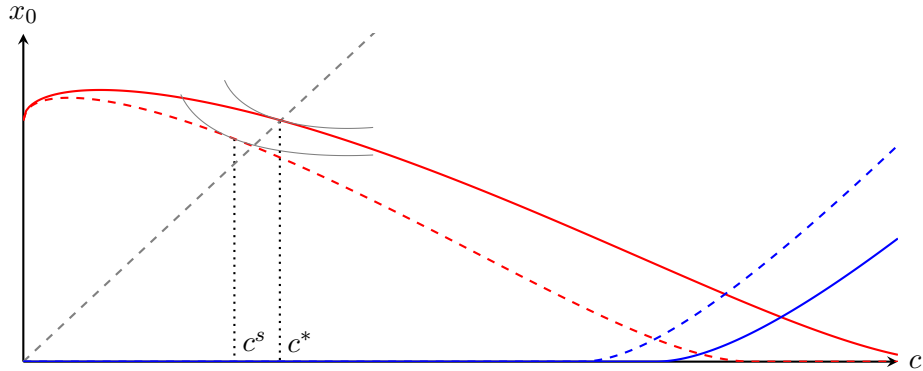
\end{example}

\section{Conclusions}
We have studied in this paper the design of optimal incentives in sequential processes by means of a stylized model of successive incentives. In our model, an initiator invests resources, for instance to recruit a follower, who generates a value, and who can do the same afterwards to continue the process. There exists a common success rate, which is increasing in the investment. Whenever a follower is recruited, it generates a value. Reward rules specify how the generated value in each step of the process is distributed, and thus induce a game among the agents. Our model is general enough to permit a wide variety of reward rules, accommodating multiple options to design successive incentives.

The main message we can take away from our results is that remarkably simple reward rules suffice to induce optimal investment behavior. Such rules treat all agents identically, except possibly for the initiator. While the designer of the reward rule has the option to construct very complicated incentive structures that may involve rewards conditional on reaching future stages in the process and treat agents very differently, this turns out to be unnecessary (albeit such rules may achieve optimal incentives as well). 

This conclusion echoes findings in contract theory where it has long been argued that real world schemes are surprisingly uniform across a wide range of circumstances, and we observe a prevalence of fairly simple contracts \cite[e.g.,][]{hart_holmstrom_1987, chiappori_salanie_2003}. The literature has also argued that simplicity may arise from a desire to offer a contract that is robust to informational assumptions \cite[e.g.,][]{bergemann2005robust, carroll2019robustness}. And, more recently, that single contracts may be optimal in environments with adverse selection and moral hazard when agents are risk-neutral and have limited liability \cite[e.g.,][]{gottlieb2022simple}.

To conclude, we also mention that our basic model is sufficiently general to allow for multiple interpretations. For instance, we stress that our model can be interpreted as one in which agents collaborate solving individual tasks sequentially. More precisely, each agent makes an effort (invests) to (independently) solve a task and generates a reward if the task is successfully solved. The higher the investment, the higher the probability of success (up to a certain common bound). And an agent is granted the chance to solve the task if and only if all predecessors successfully solve their own tasks. Then the issue is to design optimal contracts to allocate the rewards so that the overall amount of rewards obtained is maximized, or so that the expected payoff of the initial agent is maximized. This interpretation helps connect our results to recent contributions in the literature on the design of optimal investment inducing mechanisms in teamwork \cite[e.g.,][]{bernstein2012contracting, winter2004incentives, winter2006optimal, halac2022monitoring}.

\section{Appendix}

Various versions of Lemma~\ref{LE:math} are used throughout the analysis.
For the first-best investment and the socially optimal equilibrium, we use $h(c) = c$;
for the initiator-optimal equilibrium, $h(c) = g(c)$;
and for the socially optimal self-financed equilibrium, $h(c) = g(c) + c$.
The corresponding value $d \geq 0$ is defined through $h(d) = 1$.

\begin{lemma} \label{LE:math}
    Let $d \geq 0$ and $h \colon [0,d] \to [0,1]$ be differentiable, convex, and not identical to $1$.
    Then
    \[
        q(c) \equiv \frac{1 - h(c)}{1 - p(c)}
    \]
    is single-peaked on $[0,d]$.
\end{lemma}

\begin{proof}
    Upon differentiating,
    \[
        q'(c) = \frac{p'(c) \cdot ( 1 - h(c) ) - h'(c) \cdot (1 - p(c) )}{(1 - p(c))^2}.
    \]
    As $p(c) < 1$, $q$ is increasing whenever $p'(c) \cdot (1 - h(c)) - h'(c) \cdot (1 - p(c)) > 0$.
    Upon differentiating this expression:
    \begin{align*}
        & \left ( p'(c) \cdot (1 - h(c)) - h'(c) \cdot (1 - p(c)) \right )' \\
        & \hspace{3em} = p''(c) \cdot (1 - h(c) ) - p'(c) h'(c) - h''(c) \cdot (1 - p(c)) + p'(c) h'(c) \\
        & \hspace{3em} = p''(c) \cdot (1 - h(c) ) - h''(c) \cdot (1 - p(c)).
    \end{align*}
    Here, $p''(c) < 0$, $h(c) \leq 1$, $h''(c) \geq 0$, and $p(c) < 1$, so the expression is negative (except possibly at $c = 0$ or $c = d$, where it may equal zero).
    Hence, $q$ is single-peaked on $[0,d]$.
\end{proof}

\begin{lemma} \label{LE:zeroInitiator}
    For each equilibrium profile $x$ 
    such that $x_0 = 0$, there is an equilibrium profile $\tilde{x}$
    such that $\tilde{x}_0 > x_0$ and $\mathbb{W}(\tilde{x}) > \mathbb{W}(x)$.
\end{lemma}

\begin{proof}
    Let $\tilde{f} \in F$ award everything to the initiator:
    for each $j \geq 0$, $\tilde{f}(0,j) = j$.
    The associated equilibrium investment profile is $\tilde{x} = (\tilde{x}_0, 0, 0 \dots)$.
    Hence, $\mathbb{W}(\tilde{x}) = 1 + p(\tilde{x}_0) - \tilde{x}_0$.
    This is closely related to the initiator's optimization problem:
    the initiator chooses $\tilde{x}_0$ to maximize $p(\tilde{x}_0) - \tilde{x}_0$.
    Hence, we will have $p'(\tilde{x}_0) = 1$, so $\tilde{x}_0 > 0 = x_0$ and $\mathbb{W}(\tilde{x}) > 1 = \mathbb{W}(x)$.
\end{proof}

\begin{lemma} \label{LE:convex}
    Let $x = (x_0, c, c, \dots)$ 
    and $y = (y_0, c, c, \dots)$ be two equilibrium profiles, and $\lambda \in [0,1]$.
    Then, $\lambda x + (1 - \lambda) y$ is an equilibrium profile.
\end{lemma}

\begin{proof}
    Let $x$ and $y$ be two equilibrium profiles and $f^x$ and $f^y$ their corresponding supporting rules. 
    It suffices to consider $x \neq y$ and $\lambda \in (0,1)$. 
    Assume, without loss of generality, that $x_0 > y_0$.
    Fix $\lambda \in (0,1)$ and define $\tilde{x} = \lambda x + (1 - \lambda) y$, so $\tilde{x} = (\tilde{x}_0, c, c, \dots)$ with $x_0 > \tilde{x}_0 > y_0$.
    Recall that $\frac{g(\cdot)}{p(\cdot)}$ is increasing and define further 
    \[
        \alpha = 
        \frac{
        \frac{g(\tilde{x}_0)}{p(\tilde{x}_0)} - \frac{g(y_0)}{p(y_0)} 
        }
        {
        \frac{g(x_0)}{p(x_0)} - \frac{g(y_0)}{p(y_0)} 
        } \in (0,1)
    \]
    and $\tilde{f} = \alpha f^x + (1 - \alpha) f^y$.
    That is, $\tilde{f}(i,j) = \alpha f^x(i,j) + (1 - \alpha) f^y(i,j)$.
    For each agent $i > 0$,
    \[
        R_i(\tilde{x},\tilde{f}) - \tilde{f}(i,i)
        = \alpha ( R_i(x, f^x) - f^x(i,i) ) + (1 - \alpha) ( R_i(y, f^y) - f^y(i,i) )
        = \frac{g(c)}{p(c)}.
    \]
    That is, $\tilde{x}_i = c$ is optimal under $\tilde{f}$.
    It remains to verify that $\tilde{x}_0$ is optimal to agent $0$.
    We have 
    \[
        R_0(\tilde{x},\tilde{f})
        = \alpha R_0(\tilde{x},f^x) + (1 - \alpha) R_0(\tilde{x}, f^y). 
    \]
    By design, $f^x(0,0) = f^y(0,0) = \tilde{f}(0,0) = 1$.
    Hence,
    \[
        R_0(\tilde{x},\tilde{f}) - \tilde{f}(0,0)
        = \alpha ( R_0(\tilde{x},f^x) - f^x(0,0) ) + (1 - \alpha) ( R_0(\tilde{x}, f^y) - f^y(0,0) ).
    \]
    As $f^x$ and $f^y$ support $x$ and $y$ respectively, 
    $R_0(x,f^x) - f^x(0,0) = g(x_0) / p(x_0)$ and $R_0(y,f^y) - f^y(0,0) = g(y_0) / p(y_0)$. 
    Hence, as desired,
    \[
        R_0(\tilde{x},\tilde{f}) - \tilde{f}(0,0)
        = \alpha \cdot \frac{g(x_0)}{p(x_0)} + (1 - \alpha) \cdot \frac{g(y_0)}{p(y_0)}
        = \frac{g(\tilde{x}_0)}{p(\tilde{x}_0)}. \qedhere
    \]
\end{proof}

Next, we turn to the characterization result.
This can be understood as follows.
Fix the investment $c \geq 0$ for agents $i > 0$.
In \textsc{Part I} of the proof, we show that there is a lower and an upper bound, say $x_0$ and $y_0$, on agent $0$'s investment for the profile to be supportable.
In \textsc{Part II}, we construct rules to support $x = (x_0, c, c, \dots)$ and $y = (y_0, c, c, \dots)$;
by Lemma~\ref{LE:convex}, all intermediate $\tilde{x} = (\tilde{x}_0, c, c, \dots)$ can also be supported.
The construction will be based on the following four types of rules:

\begin{alignat*}{3}
    &f^1 = \left[ \begin{array}{cccc} 
        1 \\
        2 - \F & \F \\
        3 - \alpha - \F & \alpha & \F \\
        4 - 2 \alpha - \F & \alpha & \alpha & \F \\
        \vdots
    \end{array} \right] 
    \quad &&
    f^2 = \left[ \begin{array}{cccc}
        1 \\
        \alpha - \F & 2 - \alpha + \F \\
        \alpha - \F & 2 & 1 - \alpha + \F \\
        \alpha - \F & 2 & 1 & 1 - \alpha + \F \\
        \vdots 
    \end{array} \right]
    \\
    &f^3 = \left[ \begin{array}{cccc}
        1 \\
        2 - \F & \F \\
        2 - \beta - \F & 1 + \beta & \F \\
        2 - \beta - \F & 1 & 1 + \beta & \F \\
        \vdots 
    \end{array} \right]
    \quad &&
    f^4 = \left[ \begin{array}{ccccc} 
        1 \\
        \beta & 2 - \beta \\
        0 & 3 - \F & \F \\
        0 & 3 - \beta - \F & 1 + \beta & \F \\
        0 & 3 - \beta - \F & 1 & 1 + \beta & \F \\
        \vdots 
    \end{array} \right]
\end{alignat*}

\begin{table}[!htb]
    \centering
    \renewcommand{\arraystretch}{2}
    \begin{tabular}{ccc} \toprule
        $\ell$ & $R_0(x,f^\ell) - f^\ell(0,0)$ & $R_i(x,f^\ell) - f^\ell(i,i)$ \\ \midrule
        1 & $\displaystyle \frac{1 - \alpha p(c)}{1 - p(c)} - \F$ & $\alpha - \F$ \\ 
        2 & $\alpha - 1 - \F$ & $\alpha - \F$ \\ 
        3 & $1 - \beta p(c) - \F$ & $1 + \beta ( 1 - p(c) ) - \F$ \\ 
        4 & $\beta ( 1 - p(c) ) - 1$ & $1 + \beta ( 1 - p(c) ) - \F$ \\
        \bottomrule
    \end{tabular}
    \caption{Expected returns from investment for each rule $f^\ell$.
    Agent $i$ refers to $i > 0$.}

    \label{TAB:summary}
\end{table}

Table~\ref{TAB:summary} summarizes the expected returns from investment.
Specifically, for agent $0$ at $f^1$ (using that $1 + 2p + 3p^2 + \dots = 1 / (1 - p)^2$), we have
\[
    R_0(x,f^1) - f^1(0,0) = \frac{1 - \alpha p(c)}{1 - p(c)} - \F.
\]
In all other cases, we have $f^\ell(i,i+2) = f^\ell(i,i+3) = \dots$;
then 
\[
    R_i(x,f^\ell) = (1 - p(c)) f^\ell(i,i+1) + p(c) f^\ell(i,i+2). 
\]
The parameters $\alpha$ and $\beta$ will be set to meet the equilibrium conditions $R_i(x,f) - f(i,i) = g(x_i) / p(x_i)$ for $i > 0$. 
This immediately gives $\alpha = g(c) / p(c) - \F$ for $f^1$ and $f^2$, whereas
\[
    1 + \beta (1 - p(c)) - \F = \frac{g(c)}{p(c)}
    \implies
    \beta = \frac{g(c)}{p(c)} - \frac{1 - g(c) - \F}{1 - p(c)},
\]
for $f^3$ and $f^4$.
Hence, at this point we have designed the relevant rules and identified the parameter values for which the equilibrium conditions are satisfied for all agents $i > 0$.
In \textsc{Part II} of the proof, it remains to check that agent $0$ also optimizes accordingly and that the rules are well-defined (they are balanced by construction, but we must verify that values are non-negative).

\begin{lemma} \label{LE:characterization}
    Let $x = (x_0, c, c, \dots) \in X$ with $x_0 > 0$ and $\F \geq 0$.
    Then there is $f \in F$ such that $f(i,i) \geq \F$ for each $i \geq 0$, supporting $x$ 
    if and only if
    \[
        \frac{g(c)}{p(c)} + \F - 2 \leq \frac{g(x_0)}{p(x_0)} \leq \frac{1 - g(c) - \F}{1 - p(c)}.
    \]
\end{lemma}

\begin{proof}
    \textsc{Part I:} $\implies$
    Let 
    $x = (x_0, c, c, \dots)$ be an equilibrium such that $x_0 > 0$ and $f$ be its supporting rule such that $f(i,i) \geq \F$ for each $i \geq 0$.
    We aim to show that the two inequalities hold.
    As $f(i,i) \geq \F$, 
    \begin{align*}
        \hat{\mathbb{V}}(x,f) 
        &= f(0,0) + p(x_0) f(1,1) + p(x_0) p(c) f(2,2) + \dots + g(x_0) + p(x_0) g(c) + p(x_0) p(c) g(c) + \dots \\
        &\geq f(0,0) + p(x_0) f(1,1) + g(x_0) + p(x_0) ( g(c) + p(c) \F ) (1 + p(c) + p^2(c) + \dots ) \\
        &= 1 + p(x_0) f(1,1) + g(x_0) + p(x_0) ( g(c) + p(c) \F ) \mathbb{V}(c,c,\dots).
    \end{align*}
    As $f$ supports $x$, 
    by Lemma~\ref{OBS:aggregate}, $\hat{\mathbb{V}}(x,f) \leq \mathbb{V}(x) = 1 + p(x_0) \mathbb{V}(c,c,\dots)$.
    Hence,
    \[
        p(x_0) f(1,1) + g(x_0) + p(x_0) (g(c) + p(c) \F) \mathbb{V}(c,c,\dots) 
        \leq p(x_0) \mathbb{V}(c,c,\dots).
    \]
    Rearrange, divide by $p(x_0)$, and use that $\mathbb{V}(c,c,\dots) = 1 / (1 - p(c))$ and $f(1,1) \geq \F$:
    \[
        \frac{g(x_0)}{p(x_0)} 
        \leq (1 - g(c) - p(c) \F) \mathbb{V}(c,c,\dots) - f(1,1) 
        = \frac{1 - g(c) - p(c) \F}{1 - p(c)} - f(1,1)
        \leq \frac{1 - g(c) - \F}{1 - p(c)}.
    \]
    
    For the second inequality, multiply the penultimate $g(x_0) / p(x_0) \leq ( (1 - g(c) - p(c) \F) / (1 - p(c)) - f(1,1)$ by $- (1 - p(c)) < 0$ to flip the inequality:
    \begin{align*}
        - (1 - p(c)) \cdot \frac{g(x_0)}{p(x_0)}
        &\geq
        (1 - p(c)) f(1,1) + g(c) + p(c) \F - 1 \\
        \iff p(c) \cdot \frac{g(x_0)}{p(x_0)} 
        &\geq (1 - p(c)) f(1,1) + g(c) + p(c) \F - 1 + \frac{g(x_0)}{p(x_0)}.
    \end{align*}
    For the final term, the equilibrium condition asserts that $g(x_0) / p(x_0) = R_0(x,f) - f(0,0)$.
    By construction, $R_0(x,f) = (1 - p(c)) f(0,1) + p(c) (1 - p(c)) f(0,2) + \dots \geq (1 - p(c)) f(0,1)$, $f(0,0) = 1$, and $f(0,1) + f(1,1) = 2$.
    Hence, 
    \[
        \frac{g(x_0)}{p(x_0)} 
        = R_0(x,f) - f(0,0) \geq 1 - 2p(c) - (1 - p(c)) f(1,1).
    \]
    Using this expression for the final term above, we have
    \[
        p(c) \cdot \frac{g(x_0)}{p(x_0)} 
        \geq g(c) + p(c) \F - 2p(c).
    \]
    Divide by $p(c)$ to obtain the desired inequality.

    \bigskip
    \textsc{Part II:} $\impliedby$
    Let $\F \geq 0$ and $\tilde{x} = (\tilde{x}_0, c, c, \dots)$ with $\tilde{x}_0 > 0$ satisfy the inequalities of the statement.
    Then $1 - g(c) - \F > 0$.
    We will show that $\tilde{x}$ is an equilibrium profile.
    Given $c$, the two inequalities provide bounds on agent $0$'s investment:
    the lower bound is $x_0$ such that $g(x_0) / p(x_0) = \max \{ g(c) / p(c) + \F - 2, 0 \}$;
    the upper is $y_0$ such that $g(y_0) / p(y_0) = (1 - g(c) - \F) / (1 - p(c))$.
    Given that $g(\cdot) / p(\cdot)$ is increasing, we have $x_0 \leq \tilde{x}_0 \leq y_0$.
    In what follows, we construct rules to support $x = (x_0, c, c, \dots)$ and $y = (y_0, c, c, \dots)$;
    by Lemma~\ref{LE:convex}, $\tilde{x}$ is supported by a combination of these rules.
    We refer to Table~\ref{TAB:summary} for a summary of the expected investment returns $R_i(x,f^\ell) - f^\ell(i,i)$.
    Recall that we have identified the parameter values $\alpha$ and $\beta$ for which the equilibrium conditions are satisfied for all agents $i > 0$;
    it remains to check that also agent $0$ optimizes at $x_0$ and that $f(i,j) \geq 0$ everywhere.

    Below, we consider the two natural cases corresponding to the bounds (\textsc{Upper} and \textsc{Lower}), which then are further subdivided in two parts each (\textsc{Left} and \textsc{Right}) depending on whether $g(c) \leq p(c)$.
    In terms of Figure~\ref{FIG:overview}, \textsc{Upper} refers to the red dashed curve and \textsc{Lower} to the blue dashed curve (and horizontal axis);
    \textsc{Left} refers to the left-most portion of the graph left of the dashed vertical line at $c^*$.

    \bigskip
    \noindent
    \textsc{Upper:}
    Suppose that
    \[
        \frac{g(x_0)}{p(x_0)} = \frac{1 - g(c) - \F}{1 - p(c)}.
    \]
    We further subdivide in two cases.

    \begin{quote}
        \textsc{Upper left:} 
        Suppose that $g(c) / p(c) + \F \leq 1$.
        
        In this case, use the rule $f^1$ with parameter $\alpha = g(c) / p(c) + \F$.
        As $\alpha \in [0,1]$, the rule is well-defined.
        Moreover, as desired, agent $0$'s equilibrium condition is satisfied:
        \[
            R_0(x,f^1) - f^1(0,0)
            = \frac{1 - \alpha p(c)}{1 - p(c)} - \F
            = \frac{1 - g(c) - \F}{1 - p(c)} 
            = \frac{g(x_0)}{p(x_0)}.
        \]

        \textsc{Upper right:}
        Suppose that $g(c) / p(c) + \F \geq 1$.
        In this case, use the rule $f^3$ with parameter 
        \[
            \beta 
            = \frac{g(c)}{p(c)} - \frac{1 - g(c) - \F}{1 - p(c)}
            = \frac{g(c)}{p(c)} - \frac{g(x_0)}{p(x_0)}. 
        \]
        Note that 
        \[
            \frac{g(c)}{p(c)} - \frac{1 - g(c) - \F}{1 - p(c)}
            = \frac{g(c) / p(c) - g(c) - 1 + g(c) + \F}{1 - p(c)}
            = \frac{g(c) / p(c) + \F - 1}{1 - p(c)}.
        \]
        Hence, $\beta \geq 0$.
        The second underlying assumption maintained throughout the proof is that 
        \[
            \frac{g(x_0)}{p(x_0)} \geq \frac{g(c)}{p(c)} + \F - 2,
        \]
        so $\beta \leq 2 - \F \leq 2$ and the rule is well-defined.
        Then
        \begin{align*}
            R_0(x,f^3) - f^3(0,0)
            = 1 - \beta p(c) - \F
            &= 1 - p(c) \cdot \left ( \frac{g(c)}{p(c)} - \frac{1 - g(c) - \F}{1 - p(c)} \right ) - \F \\
            &= \frac{1 - g(c) - \F}{1 - p(c)} 
            = \frac{g(x_0)}{p(x_0)}.
        \end{align*}
        Hence, as desired, agent $0$'s equilibrium condition is satisfied.
    \end{quote}
    \noindent
    \textsc{Lower:} 
    Suppose that
    \[
        \frac{g(x_0)}{p(x_0)} = \max \left \{ \frac{g(c)}{p(c)} + \F - 2, 0 \right \}.
    \]
    We further subdivide in two cases.

    \begin{quote}
        \textsc{Lower right:}
        Suppose that $g(c) / p(c) + \F \geq 1$.
        In essence, we will address two cases at once here:
        if $g(c) / p(c) + \F \leq 2$, then $x_0 = 0$;
        if $g(c) / p(c) + \F \geq 2$, then $g(x_0) / p(x_0) = g(c) / p(c) + \F - 2 \geq 0$.
        Follow the case \textsc{Upper right} with the same $\beta$ but the rule $f^4$.
        By the underlying condition $g(c) / p(c) + \F - 2 \leq (1 - g(c) - \F) / (1 - p(c))$, we have
        \[
            \beta = \frac{g(c)}{p(c)} - \frac{1 - g(c) - \F}{1 - p(c)}
            \leq 2 - \F
            \iff
            f^4(1,1) = 2 - \beta \geq \F.
        \]
        Moreover,
        \[
            R_0(x,f^4) - f^4(0,0)
            = \beta (1 - p(c)) - 1 
            = \frac{g(c)}{p(c)} - g(c) - (1 - g(c) - \F) - 1
            = \frac{g(c)}{p(c)} + \F - 2.
        \]
        It is again clear that, if $g(c) / p(c) + \F \geq 2$, then agent $0$'s equilibrium condition is satisfied. 

        In the other possible case, $g(c) / p(c) + \F < 2$, we have set $x_0 = 0$ and the above is negative.
        Although the first-order condition does not bind for agent $0$, it is still correct that $x_0 = 0$ is optimal as the expected return is negative.

        \textsc{Lower left:}
        Suppose that $g(c) / p(c) + \F \leq 1$.
        In this case, $x_0 = 0$.
        Use the rule $f^2$ with parameter $\alpha = g(c) / p(c) + \F \leq 1$.
        Hence, $\alpha - \F = g(c) / p(c) \geq 0$.
        Then $R_0(x,f^2) - f^2(0,0) \leq 0$.
        As in \textsc{Lower right}, agent $0$'s first-order condition no longer binds but it is indeed optimal to set $x_0 = 0$.
    \end{quote}

    As noted at the start of \textsc{Part II}, it now follows from Lemma~\ref{LE:convex} that $\tilde{x}$ is an equilibrium profile.
\end{proof}
\newpage

\bibliographystyle{ecta}
\bibliography{bibliography}

\end{document}